\documentclass[11pt]{article}
\usepackage{amssymb,amsmath,amsthm,epsfig}
\usepackage{latexsym, enumerate}
\usepackage{eepic}
\usepackage{epic}
\usepackage{graphicx}
\usepackage{color}
\usepackage{ifpdf}
\usepackage{subfigure}
\usepackage{tikz}
\usepackage{dsfont}
\usepackage{multirow}
\usepackage{makecell}
\usepackage{algorithm}
\usepackage{bm}
\usepackage{multirow}
\usepackage{epstopdf}

\usepackage{cite}
\usepackage{hyperref}
\usepackage{xcolor}
\hypersetup{
  colorlinks,
  linkcolor={blue!80!black},
  urlcolor={blue!80!black},
  citecolor={blue!80!black}
}
\usepackage[capitalize,noabbrev]{cleveref}

\theoremstyle{plain}
\newtheorem{lem}{Lemma}[section]
\newtheorem{thm}[lem]{Theorem}
\newtheorem{cor}[lem]{Corollary}

\theoremstyle{definition}
\newtheorem{defn}{Definition}[section]

\newtheorem{con}{Condition}[section]

\theoremstyle{remark}
\newtheorem{rem}{Remark}[section]
\newtheorem{exm}{Example}[section]

\begin{document}
\title{ \large\bf Shape reconstructions by using plasmon resonances with enhanced sensitivity}

\author{
Ming-Hui Ding\thanks{Department of Mathematics, City University of Hong Kong, Kowloon, Hong Kong, China.\ \  Email: minghuiding@hnu.edu.cn}
\and
Hongyu Liu\thanks{Department of Mathematics, City University of Hong Kong, Kowloon, Hong Kong, China.\ \ Email: hongyu.liuip@gmail.com; hongyliu@cityu.edu.hk}
\and
Guang-Hui Zheng\thanks{School of Mathematics, Hunan University, Changsha 410082, China.\ \ Email: zhenggh2012@hnu.edu.cn; zhgh1980@163.com}
}

\date{}
\maketitle

\begin{center}{\bf ABSTRACT}
\end{center}\smallskip

This paper investigates the shape reconstructions of sub-wavelength objects from near-field measurements in transverse electromagnetic scattering. This geometric inverse problem is notoriously ill-posed and challenging. We develop a novel reconstruction scheme using plasmon resonances with significantly enhanced sensitivity and resolution. First, by spectral analysis, we establish a sharp quantitative relationship between the sensitivity of the reconstruction and the plasmon resonance. It shows that the sensitivity functional blows up when plasmon resonance occurs. Hence, the signal-to-noise ratio is significantly improved and the robustness and effectiveness of the reconstruction are ensured. Second, a variational regularization method is proposed to overcome the ill-posedness, and an alternating iteration method is introduced to automatically select the regularization parameters. Third, we use the Laplace approximation method to capture the statistical information of the target scattering object. Both rigorous theoretical analysis and extensive numerical experiments are conducted to validate the promising features of our method.

\smallskip
{\bf keywords:} shape reconstruction; plasmon resonance; sensitivity analysis; alternating iteration; Laplace approximation

\section{Introduction}

Surface plasmon resonance is the resonant oscillation of conducting electrons at the interface between negative and positive permittivity materials excited by proper incident waves. Plasmon technology is revolutionizing many industrial applications including biosensing, disease diagnosis, catalysis and  molecular dynamics research \cite{AHLSZV2010,BGGC,R2003,SSMS2000}; super-focusing and high-resolution imaging \cite{DLL2020,DLZ2022,PJKL}; and invisibility cloaking \cite{BB2010,ACKLM2013,CKKL2014,GWM3,LL2018,LLL2015}. In recent years, the mathematical understanding of plasmon resonances has also attracted considerable attention with much progress, in particular, their intriguing connection to the spectral theory of Neumann-Poincar\'e-type operators; see e.g. \cite{ACL2003,ACKLM2013,AKL,BLL,DLZ2,G2014,LiLiu2d} and the references cited therein.

In this paper, we study the utilization of plasmon resonances to the reconstruction of the shape of an anomalous nano-size inclusion by the associated near-field electromagnetic measurement. This geometric inverse shape problem arises in a variety of applications, especially bio-medical imaging. We shall be mainly concerned with the time-harmonic transverse magnetic (TM) scattering, though the study can be readily extended to the transverse electric (TE) case. Let $D\subset\mathbb{R}^2$ be a bounded and simply connected domain of class $C^{1,\alpha}$ for some $0<\alpha<1$, which signifies the support of an anomalous inclusion. The medium property of the inclusion $D$ is characterised by the electric permittivity {$\varepsilon_c\in\mathbb{R}_+$} and the magnetic permeability $\mu_c$, while the homogeneous matrix medium in $\mathbb{R}^2\backslash\overline{D}$ is characterized by $\varepsilon_m\in\mathbb{R}_+$ and $\mu_m\in\mathbb{R}_+$, respectively. In what follows, we set
\begin{equation}\label{eq:medn1}
\varepsilon_D=\varepsilon_m\chi(\mathbb{R}^2\backslash\overline{D})+\varepsilon_c\chi({D}), \ \ \mu_D=\mu_m\chi(\mathbb{R}^2\backslash\overline{D})+\mu_c\chi(D),
\end{equation}
where $\chi$  signifies the characteristic function. {The medium parameter $\mu_c$ varies according to the operating frequency $\omega\in\mathbb{R}_+$ of the incident wave}, which is described by the Drude model. {We shall write $\mu_c(\omega)$} to indicate such dependence and the details shall be supplemented in what follows. In principle, we shall make critical use of the result that for proper selection of the operating frequencies, {one has $\Re\mu_c<0,\Im\mu_c>0$, which is known as the negative material \cite{SPW2004,V1968,SK2000}.} Define
\begin{align*}
k_m:=\omega\sqrt{\varepsilon_m\mu_m}, \ \ k_c:=\omega\sqrt{\varepsilon_c\mu_c}.
\end{align*}
Now, we let $u^i(x) = e^{\mathrm{i}k_m d\cdot x}$, $\mathrm{i}:=\sqrt{-1}$, be a time-harmonic incident plane wave with $d\in\mathbb{S}^1$ signifying the impinging direction. The TM scattering induced by the interaction between the incident wave $u^i$ and the medium inclusion $(D; \varepsilon_c, \mu_c)$ is governed by the following PDE system:
\begin{align}
\label{sca equ}
\begin{cases}
&\displaystyle\nabla\cdot\frac{1}{\mu_D}\nabla u+\omega^2\varepsilon_D u=0\ \ \ \mathrm{in}\ \ \mathbb{R}^2\backslash\partial D,\medskip\\
&\displaystyle u_{+}=u_{-} \hspace*{2.9cm} \mathrm{on}\ \ \partial D,\medskip\\
&\displaystyle \frac{1}{\mu_m}\frac{\partial u}{\partial\nu}\bigg{|}_{+}=\frac{1}{\mu_c}\frac{\partial u}{\partial\nu}\bigg{|}_{-}\hspace*{.95cm} \mathrm{on}\  \partial D,\medskip\\
&\displaystyle u^s:=u-u^i\ \ \  \mathrm{satisfies\  the\  Sommerfeld\  radiation\  condition,}
\end{cases}
\end{align}
where the last condition signifies that
\begin{align*}
\frac{\partial u^s}{\partial |x|}-\mathrm{i}k_m u^s=O(|x|^{-\frac{3}{2}}) \ \ \mathrm{as}\ \  |x|\rightarrow +\infty,
\end{align*}
which holds uniformly in the angular variable $\hat{x}=x/|x|\in\mathbb{S}^1$. The well-posedness of the forward scattering problem \eqref{sca equ} shall be implied in our subsequent study of the associated inverse problem and there exists a unique solution $u\in H_{loc}^1(\mathbb{R}^2)$. Associated with \eqref{sca equ}, we introduce the following measurement operator $\Lambda_D$:{
 \begin{equation*}\label{eq:md1}
 \Lambda_{D}(u^i)=u^s|_{\partial\Omega}\in L^2(\partial\Omega),
 \end{equation*}
where $u^s$ is the scattering field to \eqref{sca equ} and $\Omega$ is a smooth domain containing $D$}. Without loss of generality, we assume that $\Omega$ is a central ball of radius $R_0\in\mathbb{R}_+$ throughout the rest of the paper. In this paper, we are mainly concerned with the geometric inverse problem of recovering $D$ by knowledge of $\Lambda_D(u^i)$ associated with a single incident wave $u^i$. By introducing an abstract operator $\mathcal{F}$ defined by \eqref{sca equ} which sends the inclusion $D$ to the measurement data, the inverse problem can be recast as the following operator equation:
 \begin{equation}\label{eq:ip1}
 \mathcal{F}(D)=\Lambda_D(u^i)\quad \mbox{with a fixed $u^i$}.
 \end{equation}
Two remarks are in order regarding the inverse problem \eqref{eq:ip1}. First, it can be readily verified that the inverse problem \eqref{eq:ip1} is nonlinear. Moreover, we are mainly interested in the case that $\omega\cdot\mathrm{diam}(D)\ll 1$ (it is actually $k_m\cdot\mathrm{diam}(D)\ll 1$ since we shall normalize $\varepsilon_m$ and $\mu_m$ in what follows), i.e. the size of the anomalous inclusion $D$ is much smaller compared to the operating wavelength $2\pi/\omega$. This is known as the sub-wavelength scale or the quasi-static regime. The scattering information in the quasi-static regime is very weak, and in the presence of measurement noise, the signal-to-noise ratio is low \cite{Lim2012,ACL2022}. This makes the inverse problem \eqref{eq:ip1} notoriously ill-posed. Moreover, due to the diffraction limit, the fine detail information of $D$ carried in the scattered field is restricted to the vicinity of the inclusion, resulting in usually low-resolution reconstructions of \eqref{eq:ip1}. Our current study is motivated by addressing those challenges by using plasmon resonances. Second, in \eqref{eq:ip1} {we assume that $\varepsilon_c$ and $\mu_c(\omega)$} are a-priori known. In the practical setup, this corresponds to that we know the material property of the medium inclusion $D$, but we do not know its location and shape and intend to reconstruct them by the scattering method. In principle, our method can be extended to dealing with the case of simultaneously recovering $D$ and its medium content by properly augmenting the optimization functional in Section~\ref{sect:4}. However, the focus of the current study is to theoretically and numerically verify that by using plasmon resonances, one can have a much more stable and high-resolution scheme of reconstructing $D$. In order to have a focusing theme, we choose to study the simultaneous recovery problem in a forthcoming paper.

As mentioned earlier, we shall make use of plasmon resonances for tackling the inverse problem \eqref{eq:ip1}. In the practical setup, one observes a significant enhancement of the scattering field (say e.g. in terms of the scattering amplitude) for certain specific operating frequencies, which corresponds to the occurrence of plasmon resonances. In fact, a mathematical framework of applying plasmon resonances for shape reconstructions was initiated in a recent article \cite{DLZ2022}.  Next, we highlight several significant novelties and technical developments of the current study compared to that in the aforementioned article. First, the article \cite{DLZ2022} addresses the shape reconstructions in electrostatics, i.e. the limiting case with zero frequency, whereas we consider the wave scattering in the quasi-static regime. It turns out that the presence of the frequency leads to tremendous technical difficulties in the relevant asymptotic analysis of sharply quantifying the relationship between the shape sensitivity and the plasmon resonance; see Section~\ref{sect:3.2}. Moreover, we derive a much more thorough understanding of the shape sensitivity, especially the connections between the size and curvature of the anomalous inclusion, which were not considered in \cite{DLZ2022}. Second, there is a particular point worth special emphasis. It is known that the plasmon resonance is connected to the spectrum of the Neumann-Poincar\'e (NP) operator; see also Section~\ref{sect:pr1} in what follows. For a generic non-radial domain, $0$ belongs to the essential spectrum of the underlying NP operator. The analysis in \cite{DLZ2022} was mainly conducted around an NP eigenvalue, i.e. the discrete spectrum of the NP operator, and the case around $0$ was not considered due to significant technical difficulties. In this paper, we present a subtle analysis of the shape sensitivity around $0$ in the radial case. Our result shows that in such a case one can have much enhanced sensitivity of the reconstruction compared to the non-zero case. Though $0$ is an NP eigenvalue for a radial domain, our result in this aspect provides unobjectionable implications that using the plasmon resonance associated with the essential spectrum of the underlying NP operator may yield enhanced reconstructions. Third, in the numerical implementation, we combine the variational regularization method with the Laplace approximation (LA) method to solve the inverse problem. This hybrid method can efficiently calculate the minimization point (MAP point) as well as effectively capture the uncertainty information of the solution. Compared with the reconstruction scheme in \cite{DLZ2022}, an alternating iterative approach is developed based on a hierarchical Bayesian framework. This leads to an automatic selection of the regularization parameters, and thus makes the reconstruction scheme practically more appealing.

The rest of the paper is organized as follows. In Section 2, we introduce some preliminary and auxiliary results on layer potential theory and plasmon resonances which shall be needed for the subsequent analysis. In Section 3, we derive spectral expansions of the shape sensitivity functionals and perform sensitivity analysis for perturbed domains. Section 4 is devoted to the hierarchical Bayesian model, alternating iteration method as well as the LA method. In Section 5, we conduct extensive numerical experiments to validate the theoretical findings. The paper is concluded in Section 6 with some relevant discussions.

\section{Preliminary and auxiliary results}

\subsection{Layer potential theory}

We introduce some preliminary knowledge on the layer potential operators by following the treatment in \cite{Millien2017,DLZ2021} which shall be needed in our subsequent analysis. The fundamental solution $G$ to the Helmholtz operator $\Delta+k^2$ in $\mathbb{R}^2$ is given by
\begin{align*}
G(x,y,k)=-\frac{\mathrm{i}}{4}H^{(1)}_0(k|x-y|),\  x\neq y,
\end{align*}
where $H^{(1)}_0$ is the Hankel function of the first kind and order 0. Let $\mathcal{S}_D^k$ and $\mathcal{D}_D^k$ be the single- and double-layer potentials defined by
\begin{align*}
&\mathcal{S}_D^k[\psi](x)=\int_{\partial D}G(x,y,k)\psi(y)d\sigma(y),\ \ \ \ \ x\in \mathbb{R}^2,\\
&\mathcal{D}_D^k[\psi](x)=\int_{\partial D}\frac{\partial G(x,y,k)}{\partial \nu(y)}\psi(y)d\sigma(y),\ \ x\in \mathbb{R}^2\backslash\partial D\nonumber
\end{align*}
for some surface density $\psi\in L^2(\partial D)$. The following jump relations hold for these operators across the boundary $\partial D$:
\begin{align*}
\frac{\partial \mathcal{S}_D^k[\psi]}{\partial \nu}\bigg{|}_{\pm}(x)=\bigg{(}\pm\frac{1}{2}Id+(\mathcal{K}_D^k)^*\bigg{)}[\psi](x), \ \ \
\mathcal{D}_D^k[\psi]\bigg{|}_{\pm}(x)=\bigg{(}\mp\frac{1}{2}Id+\mathcal{K}_D^k\bigg{)}[\psi](x),
\end{align*}
where $Id$ indicates the identity operator, $\pm$ signifies the limits taken from the inside and outside of $D$, respectively, and $\mathcal{K}_D^k$ is the adjoint operator of the Neumann-Poincar\'e operator $(\mathcal{K}_D^k)^*$:
\[
(\mathcal{K}_D^k)^*[\psi](x)=\int_{\partial D}\frac{\partial G(x,y,k)}{\partial \nu(x)}\psi(y)d\sigma(y),\ \ x\in\partial D.
\]
In the sequel, we set $\mathcal{S}_D^0=\mathcal{S}_D$, $\mathcal{D}_D^0=\mathcal{D}_D$ and $(\mathcal{K}_D^0)^*=\mathcal{K}_D^*$.

In $\mathbb{R}^2$, the operator $\mathcal{S}_D: {H}^{-1/2}(\partial D)\rightarrow {H}^{1/2}(\partial D)$ is not invertible. We define the following substitute for $\mathcal{S}_D$ to avoid the non-invertibility:
\begin{align*}
\widetilde{\mathcal{S}}_D[\psi]=
\begin{cases}
 {\mathcal{S}}_D[\psi],\ \ \ \ \ \mathrm{if}\ \langle\psi,\chi(\partial D)\rangle_{-\frac{1}{2},\frac{1}{2}}=0,\\
-\chi(\partial D),\ \ \mathrm{if}\  \psi=\varphi_0,
\end{cases}
\end{align*}
where $\varphi_0$ is the unique eigenfunction of $\mathcal{K}_D^*$ of eigenvalue 1/2 such that $\langle\varphi_0,\chi(\partial D)\rangle_{-\frac{1}{2},\frac{1}{2}}=1$. Furthermore, by using the Calder\'on identity: $\mathcal{K}_D\widetilde{\mathcal{S}}_D=\widetilde{\mathcal{S}}_D \mathcal{K}^*_D$, we can define a new inner product
\begin{align}\label{inner}
\langle u,v\rangle_{\mathcal{H}^*}=-\langle u,\widetilde{\mathcal{S}}_D[v]\rangle_{-\frac{1}{2},\frac{1}{2}},
\end{align}
where $\langle\cdot,\cdot\rangle_{-\frac{1}{2},\frac{1}{2}}$ is the duality pairing between ${H}^{-1/2}(\partial D)$ and ${H}^{1/2}(\partial D)$, and $\mathcal{H}^*$ is equivalent to ${H}^{-1/2}(\partial D)$.

For the fundamental solution $G$, we shall need the following behavior of the Hankel function near 0:
\begin{align}\label{eq:nn1}
-\frac{\mathrm{i}}{4}H^{(1)}_0(k|x-y|)=\frac{1}{2\pi}\log(|x-y|)+\tau_k+\sum_{j=1}^{\infty}(b_j\log(k|x-y|)+c_j)(k|x-y|)^{2j},
\end{align}
where
\begin{align}\label{eq:nn2}
&\tau_k=\frac{1}{2\pi}(\log k+\gamma-\log 2)-\frac{\mathrm{i}}{4},\  b_j=\frac{(-1)^j}{2\pi}\frac{1}{2^{2j}(j!)^2},\ c_j=-b_j\bigg{(}\gamma-\log 2-\frac{\mathrm{i}\pi}{2}-\sum_{n=1}^{j}\frac{1}{n}\bigg{)},
\end{align}
with $\gamma$ being the Euler constant. Using \eqref{eq:nn1} and \eqref{eq:nn2}, one can readily derive that the single-layer potential operator has the following asymptotic expansion which converges in  $\mathcal{L}(\mathcal{H}^*(\partial D),\mathcal{H}(\partial D))$:
\begin{align}\label{series}
\mathcal{S}_D^k=\widehat{\mathbb{S}}_D^k+\sum_{j=1}^{\infty}(k^{2j}\log k)\mathbb{S}_{D,j}^{(1)}
+\sum_{j=1}^{\infty}k^{2j}\mathbb{S}_{D,j}^{(2)},\ \ x\in \partial D,
\end{align}
where
\begin{align}
\label{Sk}
&\widehat{\mathbb{S}}_D^k[\psi](x)=\widetilde{\mathcal{S}}_D[\psi](x)+\Upsilon_k[\psi](x); \  \Upsilon_k[\psi]=(\psi,\varphi_0)_{\mathcal{H}^*}(\mathcal{S}_D[\varphi_0]+\chi(\partial D)+\tau_k),\\
\label{SDj}
&\mathbb{S}_{D,j}^{(1)}[\psi](x)=\int_{\partial D}b_j|x-y|^{2j}\psi(y)d\sigma(y),\\
\label{SD2j}
&\mathbb{S}_{D,j}^{(2)}[\psi](x)=\int_{\partial D}|x-y|^{2j}(b_j \log| x-y|+c_j)\psi(y)d\sigma(y).
\end{align}
Similarly, we have the following asymptotic expansion for the boundary integral integral operator $(\mathcal{K}^k_D)^*$:
\begin{align}\label{K^*}
(\mathcal{K}^k_D)^*=\mathcal{K}_D^*+\sum_{j=1}^{\infty}(k^{2j}\log k)\mathbb{K}^{(1)}_{D,j}+\sum_{j=1}^{\infty}k^{2j}\mathbb{K}^{(2)}_{D,j},
\end{align}
where
\begin{align*}
\mathbb{K}^{(1)}_{D,j}=\int_{\partial D}b_j\frac{\partial|x-y|^{2j}}{\partial \nu(x)}\psi(y) d\sigma(y),\
\mathbb{K}^{(2)}_{D,j}=\int_{\partial D}\frac{\partial(|x-y|^{2j}(b_j\log|x-y|+c_j))}{\partial \nu(x)}\psi(y) d\sigma(y).
\end{align*}
The series \eqref{K^*} is convergent in $\mathcal{L}(\mathcal{H}^*(\partial D),\mathcal{H}^*(\partial D))$.

The following lemma summarizes some results about the Neumann-Poincar\'e operator $\mathcal{K}^*_D$ which can be conveniently found in \cite{Millien2017,DLZ2021}.
\begin{lem}\label{K basic}
(i) The operator $\mathcal{K}_{ D}^*$ is compact and self-adjoint in the Hilbert space $\mathcal{H}^*(\partial D)$ with the inner product \eqref{inner}.\\
(ii) Let $(\lambda_j,\varphi_j), j=0,1,2,...$ be the eigenvalue and normalized eigenfunction pair of $\mathcal{K}_D^*$ in $\mathcal{H}^*(\partial D)$, then $\lambda_j\in(-\frac{1}{2},\frac{1}{2}]$ and $\lambda_j\rightarrow 0$ as $j\rightarrow\infty$;\\
(iii) For any $\psi\in H^{-1/2}(\partial D)$, we have $\mathcal{K}^*_{ D}[\psi]=\sum_{j=0}^{\infty}\lambda_j\langle\psi,\varphi_j\rangle_{\mathcal{H}^*(\partial D)}\varphi_j$.
\end{lem}

\subsection{Plasmon resonances}\label{sect:pr1}

We introduce the mathematical framework of plasmon resonances. First, via the layer potential theory, the solution of (\ref{sca equ}) can be represented as
\begin{align*}
u(x)=
\begin{cases}
\displaystyle u^i(x)+\mathcal{S}^{k_m}_D[\psi](x),\hspace*{1.45cm} x\in \mathbb{R}^2\backslash{\overline{D}},\medskip\\
\displaystyle \mathcal{S}^{k_c}_D[\phi](x),\hspace*{2.9cm} x\in D,\\
\end{cases}
\end{align*}
where $(\psi,\phi)\in {H}^{-\frac{1}{2}}(\partial D)\times{H}^{-\frac{1}{2}}(\partial D)$ satisfy the following integral system
\begin{align}
\label{integ_sys}
\begin{cases}
\displaystyle\mathcal{S}_D^{k_m}[\psi]-\mathcal{S}_D^{k_c}[\phi]=-u^i\hspace*{7.1cm} \mathrm{on}\ \partial D,\medskip\\
\displaystyle\frac{1}{\mu_m}\bigg{(}\frac{1}{2}Id+(\mathcal{K}_D^{k_m})^*\bigg{)}[\psi]+\frac{1}{\mu_c}\bigg{(}\frac{1}{2}Id-(\mathcal{K}_D^{k_c})^*\bigg{)}[\phi]=-\frac{1}{\mu_m}\frac{\partial u^i}{\partial \nu}\hspace*{.56cm} \ \mathrm{on}\ \partial D.
\end{cases}
\end{align}
When $w$ is small enough, $\mathcal{S}_D^{k_c}$ is invertible \cite{DLU2017}. Therefore, from the first equation of \eqref{integ_sys}, we have
$\phi=(\mathcal{S}_D^{k_c})^{-1}(\mathcal{S}_D^{k_m}[\psi]+u^i)$. Then, by using the second equation of \eqref{integ_sys}, we can obtain that $\psi$ satisfies the following equation
\begin{align}\label{psi}
\mathcal{A}_D(\omega)[\psi]=f,
\end{align}
where
\begin{align*}
\mathcal{A}_D(\omega)&=\frac{1}{\mu_m}\bigg{(}\frac{1}{2}Id+\big{(}\mathcal{K}_D^{k_m}\big{)}^*\bigg{)}+\frac{1}{\mu_c}\bigg{(}
\frac{1}{2}Id-\big{(}\mathcal{K}_D^{k_c}\big{)}^*\bigg{)}\big{(}\mathcal{S}_D^{k_c}\big{)}^{-1}\mathcal{S}_D^{k_m},\\
f&=-\frac{1}{\mu_m}\frac{\partial u^i}{\partial \nu}+\frac{1}{\mu_c}\bigg{(}\frac{1}{2}Id-\big{(}\mathcal{K}_D^{k_c}\big{)}^*\bigg{)}\big{(}\mathcal{S}_D^{k_c}\big{)}^{-1}[-u^i].
\end{align*}
Clearly,
\begin{align}\label{AD0}
\mathcal{A}_D(0)=\mathcal{A}_{D,0}=\frac{1}{\mu_m}\bigg{(}\frac{1}{2}Id+\mathcal{K}_D^*\bigg{)}+\frac{1}{\mu_c}\bigg{(}\frac{1}{2}Id-\mathcal{K}_D^*\bigg{)}
=\bigg{(}\frac{1}{2\mu_m}+\frac{1}{2\mu_c}\bigg{)}Id-\bigg{(}\frac{1}{\mu_c}-\frac{1}{\mu_m}\bigg{)}\mathcal{K}_D^*.
\end{align}
From the spectral expansion of $\mathcal{K}_D^*$ in Lemma \ref{K basic}, it can be seen that
\begin{align*}
\mathcal{A}_{D,0}[\psi]=\sum_{j=0}^{\infty}\tau_j(\psi, \varphi_j)_{\mathcal{H}^*}\varphi_j,\ \ \ \tau_j=\frac{1}{2\mu_m}+\frac{1}{2\mu_c}-(\frac{1}{\mu_c}-\frac{1}{\mu_m})\lambda_j.
\end{align*}

In order to solve the equation \eqref{psi}, we first introduce the following definition of the index set of resonance.
\begin{defn}
We call $J\in \mathbb{N}$ the index set of resonance if $\tau_j$ is close to zero when $j\in J$ and is bounded from below when $j\in J^c$. More precisely, we choose a threshold number $\eta_0>0$ independent of $\omega$ such that $|\tau_j|\geq \eta_0>0, \ \ \mathrm{for}\ j\in J^c$.
\end{defn}
It is worth noting that for $j=0$, i.e., $\lambda_0=1/2$, we can obtain $\tau_0=1/\mu_m=O(1)$. Thus, the index set $J$ excludes 0. Moreover, we impose the following three mild conditions throughout the paper.
\begin{con}\label{1}
(i) We assume that $\omega$ and $\Im\mu_c$ are of order $o(1)$.\\
(ii) For $j\in J$, the eigenvalue $\lambda_j$ is a simple eigenvalue of the operator $\mathcal{K}_D^*$.\\
(iii) Suppose that $\mu_c\neq-\mu_m$.
\end{con}

Next, we collect the following  useful asymptotic expansion results with respect to small $\omega$, the details of which can be found in \cite{Millien2017}.
\begin{lem} For $\omega\ll 1$, the following asymptotic expansion results hold.\\
(i) The operator ${\mathcal{S}}_D^k: \mathcal{H}^*(\partial D)\rightarrow\mathcal{H}(\partial D)$  is invertible,
and
\begin{align}\label{invS}
({\mathcal{S}}_D^k)^{-1}&=\mathcal{L}_D+\mathcal{U}_k-k^2\log k \mathcal{L}_D \mathbb{{S}}_{D,1}^{(1)}\mathcal{L}_D+O(k^2)
\end{align}
with $\mathcal{L}_D=\mathcal{P}_{\mathcal{H}^*_0}\widetilde{\mathcal{S}}_D^{-1}$,
 $\mathcal{U}_k=-\frac{\langle\widetilde{\mathcal{S}}_D^{-1}[\cdot], \varphi_0\rangle
_{\mathcal{H}^*}}{\mathcal{S}_D[\varphi_0]+\tau_k}\varphi_0$, and $\mathbb{S}_{D,1}^{(1)}$ given by (\ref{SDj}). Here $\mathcal{P}_{\mathcal{H}^*_0}$ is the orthogonal projection onto $\mathcal{H}^*_0(\partial D)$, and $\mathcal{H}^*_0(\partial D)$ is the zero mean subspace of $\mathcal{H}^*(\partial D)$.\\
(ii) The operator $\mathcal{A}_D{(\omega)}$ can be expanded as follows:
\begin{align*}
\mathcal{A}_D{(\omega)}=\mathcal{A}_{D,0}+\omega^2(\log\omega)\mathcal{A}_{D,1}+O(\omega^2),
\end{align*}
 where $\mathcal{A}_{D,0}$ is given by (\ref{AD0}), and
\begin{align}\label{AD1}
\mathcal{A}_{D,1}=\mathbb{K}_{D,1}^{(1)}(\varepsilon_mId-\varepsilon_c\mathcal{P}_{\mathcal{H}^*_0})
+\frac{1}{\mu_c}\bigg{(}\frac{1}{2}Id-\mathcal{K}_D^*\bigg{)}\widetilde{\mathcal{S}}_D^{-1}\mathbb{S}_{D,1}^{(1)}
(\mu_m\varepsilon_mId-\mu_c\varepsilon_c\mathcal{P}_{\mathcal{H}^*_0}).
\end{align}
(iii) The asymptotic formulas for the eigenvalues $\tau_{j}(\omega)$ and eigenfunctions $\varphi_{j}(\omega)$ of the $\mathcal{A}_D{(\omega)}$ operator with respect to the $\omega$ are given by
\begin{align}\label{per lam}
\tau_{j}(\omega)=\tau_j+(\omega^2\log\omega)\tau_{j,1}+O(\omega^2),\ \
\varphi_{j}(\omega)=\varphi_j+(\omega^2\log\omega)\varphi_{j,1}+O(\omega^2),
\end{align}
where
\begin{align*}
\tau_{j,1}=\langle \mathcal{A}_{D,1}\varphi_j,\varphi_l  \rangle_{\mathcal{H}^*(\partial D)},\ \
\varphi_{j,1}=\sum_{j\neq l}\frac{\langle \mathcal{A}_{D,1}\varphi_j,\varphi_l\rangle_{\mathcal{H}^*(\partial D)}}{(\frac{1}{\mu_m}-\frac{1}{\mu_c})(\lambda_j-\lambda_l)}\varphi_l
\end{align*}
with $\mathcal{A}_{D,1}$ defined in \eqref{AD1}.\\
(iv)
If Condition \ref{1} is satisfied, then the scattering field $u^s=u-u^i$ admits the following representation:
\[
u^s=\mathcal{S}^{k_m}_D[\psi],
\]
where
\begin{align}\label{den}
\psi=\sum_{j\in J}\frac{\mathrm{i}\omega\sqrt{\varepsilon_m\mu_m}\langle d\cdot\nu,\varphi_j\rangle_{\mathcal{H}^*}\varphi_j+O(\omega^3\log \omega)}{\lambda-\lambda_j+O(\omega^2\log \omega)}+O(\omega)
\end{align}
with $\lambda=\frac{\mu_m+\mu_c}{2(\mu_m-\mu_c)}$.
\end{lem}

In this paper, we assume that $\mu_c$ of the nanoparticle varies with the frequency $\omega$. Then $\mu_c(\omega)$ can be described by Drude model (cf. \cite{Sarid2010}),
\[
\mu_c(\omega)=\mu_0\bigg{(}1-F\frac{\omega^2}{\omega^2-\omega_0^2+\mathrm{i}\tau^{-1}\omega}\bigg{)},
\]
where $\omega_0$ is the localized plasmon resonant frequency, $\tau\in\mathbb{R}_+$ denotes the bulk electron relaxation rate of the nanoparticle, and $F$ is the filling factor. When
\[
(1-F)(\omega^2-\omega_0^2)^2-F\omega_0^2(\omega^2-\omega_0^2)+\tau^{-2}\omega^2<0,
\]
it can be deduced that $\Re(\mu_c(\omega))<0$. Next, we introduce the following definition.
\begin{defn}
The quasi-static plasmon resonance is defined by $\omega$ such that
\begin{align}\label{plasmon}
\Re(\lambda(\omega))=\lambda_j,\ \mathrm{for}\ j\in J,
\end{align}
where $\lambda_j$ is an eigenvalue of the Neumann-Poincar\'e operator $\mathcal{K}^*_D$, and $\lambda(\omega)=\frac{\mu_m+\mu_c(\omega)}{2(\mu_m-\mu_c(\omega))}$.
\end{defn}
\begin{rem}
From the results of the spectral expansion of $\psi$ in (\ref{den}),
it is clear that $\psi$ is amplified when the plasmon resonance condition is reached, thereby the scattering field $u^s$ exhibits resonant behavior.
\end{rem}

\section{Shape sensitivity analysis}
\subsection{Sensitivity analysis for perturbed domains}
In this subsection, we first discuss the asymptotic expansion of the single-layer potential and the Neumann-Poincar\'e operator with respect to the boundary perturbation. Then we focus on the sensitivity of the perturbation domain, i.e., the effect of changes in the domain on the near-field measurement data.

Let $ X(t): [a, b]\rightarrow \mathbb{R}^2$ be the arclength parametrization of $\partial D$. Then $X\in C^2[a,b]$ satisfies $ |X'(t)| = 1$, and
\[
\partial D:=\{x=X(t), t\in[a,b]\}
\]
with $X'(t)=T(x)$ and $X''(t)=\tau(x)\nu(x)$, where $\nu$ is the outward unit normal vector field on $\partial D$. We denote the tangential derivative by $\frac{d}{d t}$. Letting $\phi \in C^2[a, b]$, we have
\begin{align*}
\frac{d^2}{d t^2}\phi(x)=\frac{\partial ^2 \phi}{\partial T^2}(x)+\tau\frac{\partial \phi}{\partial\nu}(x).
\end{align*}
The Helmholtz operator on the neighborhood of $\partial D$ can be denoted as follows:
\begin{align}\label{Helm opra}
\Delta+k^2=\frac{\partial^2}{\partial\nu^2}+\frac{\partial^2}{\partial T^2}+k^2=\frac{\partial^2}{\partial\nu^2}-\tau\frac{\partial }{\partial\nu}+\frac{d^2}{d t^2}+k^2 \ \ \ \mathrm{on}\ \partial D.
\end{align}
For $h\in C^1(\partial D)$ and a small $\epsilon$, we let $D_\epsilon$ be a deformation of $D$ given by
\[
\partial D_\epsilon:=\{ \tilde{x}=x+\epsilon h(x)\nu(x)\mid x\in \partial D\},
\]
and set $\Psi_\epsilon(x)=x+\epsilon h(x) \nu(x)$ be the diffeomorphism from $\partial D$ to $\partial D_\epsilon$. The outward unit normal $\tilde{\nu}(\tilde{x})$ and line element $d\tilde{\sigma}$ of $\partial D_\epsilon$, can be expanded uniformly as \cite{Lim2012}:
\begin{align*}
\tilde{\nu}(\tilde{x})&=\nu(x)-\epsilon h'(t)T(x)+O(\epsilon^2),\\
d\tilde{\sigma}(\tilde{x})&=d\sigma(x)-\epsilon\tau(x)h(x)d\sigma(x)+O(\epsilon^2).
\end{align*}
For $u^i(\tilde{x}), \frac{\partial u^i(\tilde{x})}{\partial \tilde{\nu}(\tilde{x})}$, by the Taylor expansion, we have
\begin{align}
{u}^i(\tilde{x})&=u^i(x)+\epsilon h(x)\frac{\partial u^i}{\partial \nu}(x)+O(\epsilon^2),\hspace*{4.05cm} \ \ x\in \partial D,\medskip\\
\frac{\partial {u}^i(\tilde{x})}{\partial \tilde{\nu}(\tilde{x})}&=\frac{\partial u^i(x)}{\partial \nu(x)}+\epsilon \bigg{(}h(x)\frac{\partial^2 u^i(x)}{\partial\nu^2(x)}-h'(x)\frac{\partial u^i(x)}{\partial T(x)}\bigg{)}+O(\epsilon^2),\hspace*{.5cm} \ \ x\in \partial D.
\end{align}

In the following lemma, we give asymptotic expansion results for the layer potential operators  $\mathcal{S}_{D_\epsilon}^k, (\mathcal{K}_{D_\epsilon}^k)^*$ and the density functions $\phi_\epsilon, \psi_\epsilon$ with respect to $\epsilon$. These results provide the basis for the analysis of shape sensitivity.
\begin{lem}
\text{\cite{Zribi2013}} Let $\phi:=\tilde{\phi}\circ \Psi_\epsilon$. There exists $C$ depending only on $N$, $\parallel X\parallel_{\mathcal{C}^2}$, and $\parallel h\parallel_{\mathcal{C}^1} $ such that
$$\hspace{-4.0cm}(i)\ \ \ \ \ \ \ \ \ \ \ \ \ \ \ \ \parallel \mathcal{S}_{D_\epsilon}^k[\tilde{\phi}]\circ \Psi_\epsilon-\mathcal{S}_{D}^k[\phi]-\epsilon\mathcal{S}_{D,k}^{(1)}[\phi]\parallel_{L^2(\partial D)}\leq C \epsilon^{2}\parallel\phi\parallel_{L^2(\partial D)},$$
where the operator $\mathcal{S}_{D,k}^{(1)}$ defined for any $\phi\in L^2(\partial D)$ by
\begin{align}\label{Sdk1}
\mathcal{S}_{D,k}^{(1)}[\phi](x)&=-\mathcal{S}_D^k[\tau h \phi](x)+h(\mathcal{K}_D^k)^*[\phi](x)+\mathcal{K}_D^k[h\phi](x)\nonumber,\\
&=-\mathcal{S}_D^k[\tau h \phi](x)+\bigg{(} h\frac{\partial(\mathcal{S}_D^k[\phi])}{\partial \nu}+\mathcal{D}_D^k[h\phi]\bigg{)}\bigg{|}_{\pm}(x),\ x\in\partial D.
\end{align}
$$\hspace{-2.3cm}(ii)\ \ \ \ \ \ \ \ \ \ \ \ \ \ \ \ \ \parallel((\mathcal{K}^k_{D_\epsilon})^*[\tilde{\phi}])\circ\Psi_\epsilon-(\mathcal{K}^k_{D})^*[\phi]
-\epsilon\mathcal{K}_{D,k}^{(1)}[\phi]\parallel_{L^2(\partial D)}\leq C \epsilon^{2}\parallel\phi\parallel_{L^2(\partial D)},$$
where
\begin{align}\label{k1}
\mathcal{K}_{D,k}^{(1)}[\phi](x)&=\tau(x)h(x)(\mathcal{K}_D^k)^*[\phi](x)-(\mathcal{K}_D^k)^*[\tau h\phi](x)+\frac{\partial (\mathcal{D}_D^k[h\phi])}{\partial \nu}(x)\nonumber\\
&-\frac{d}{d t}\bigg{(}h\frac{d(\mathcal{S}_D^k[\phi])}{dt}\bigg{)}(x)-k^2h(x)\mathcal{S}_D^k[\phi](x)\nonumber,\\
&=\bigg{(}\tau h\frac{\partial(\mathcal{S}^k_D[\phi])}{\partial \nu}-\frac{\partial(\mathcal{S}^k_D[\tau h\phi])}{\partial \nu} \bigg{)}\bigg{|}_{\pm}(x)+\frac{\partial (\mathcal{D}_D^k[h\phi])}{\partial \nu}(x)\nonumber\\
&-\frac{d}{dt}\bigg{(}h\frac{d(\mathcal{S}_D^k[\phi])}{dt}\bigg{)}(x)-k^2h(x)\mathcal{S}^k_D[\phi](x), \ \ x\in \partial D.
\end{align}
$(iii)$ Let $(\phi_\epsilon, \psi_\epsilon)$ be the solution of (\ref{integ_sys}) with $D_\epsilon$, then it satisfies
\begin{align*}
\parallel \phi_\epsilon\circ \Psi_\epsilon-\phi-\epsilon^n\phi^{(1)}\parallel_{L^2(\partial D)}+\parallel \psi_\epsilon\circ \Psi_\epsilon-\psi-\epsilon^n\psi^{(1)}\parallel_{L^2(\partial D)}\leq C\epsilon^{2},
\end{align*}
where $(\psi^{(1)}, \phi^{(1)})$ as the solution to the following system on $\partial D:$
\begin{align}
\small
\label{phipsin}
\begin{cases}
\displaystyle\mathcal{S}_D^{k_m}[\psi^{(1)}]-\mathcal{S}_D^{k_c}[\phi^{(1)}]=-h\frac{\partial u^i}{\partial \nu}+\mathcal{S}_{D,k_c}^{(1)}[\phi^{(0)}]
-\mathcal{S}_{D,k_m}^{(1)}[\psi^{(0)}],\medskip\\
\displaystyle\frac{1}{\mu_m}\bigg{(}\frac{\partial (\mathcal{S}_D^{k_m}\psi^{(1)})}{\partial \nu}\bigg{)}\bigg{|}_{+}-
\frac{1}{\mu_c}\bigg{(}\frac{\partial (\mathcal{S}_D^{k_c}\phi^{(1)})}{\partial \nu}\bigg{)}\bigg{|}_{-}=
-\frac{1}{\mu_m}\bigg{(}h\frac{\partial^2 u^i}{\partial\nu^2}-h'\frac{\partial u^i}{\partial T}\bigg{)}+\frac{1}{\mu_c}\mathcal{K}_{D,k_c}^{(1)}[\phi^{(0)}]-
\frac{1}{\mu_m}\mathcal{K}_{D,k_m}^{(1)}[\psi^{(0)}]
\end{cases}
\end{align}
with $\mathcal{S}_{D,k}^{(1)}$ and $\mathcal{K}_{D,k}^{(1)}$ are defined by \eqref{Sdk1} and \eqref{k1}, respectively.
\end{lem}

When $D$ is replaced by $D_\epsilon$ in problem (\ref{sca equ}), we apply the filed expansion (FE) method (cf.\cite{CGHIR1999}) to export asymptotic expansions of $u_\epsilon$. Firstly, we expand $u_\epsilon$ with respect to $\epsilon$ of the form,
\begin{align}\label{uFE}
u_\epsilon(x)=u_0+\epsilon u_1(x)+O(\epsilon^{2}),\ \ x\in \mathbb{R}^2.
\end{align}
Note that  $u_0=u$, and $u_1$ is well-defined in $\mathbb{R}^2\backslash \partial D$
\begin{align}
\label{un}
\begin{cases}
&\displaystyle(\Delta +\omega^2\varepsilon_c\mu_c) u_1=0\ \hspace*{.8cm} \mathrm{in}\   D,\medskip\\
&\displaystyle(\Delta +\omega^2\varepsilon_m\mu_m) u_1=0\ \hspace*{0.55cm} \mathrm{in}\  \mathbb{R}^2\backslash \overline{D},\medskip\\
&\displaystyle u_1\ \ \mathrm{satisfies\  the\  Sommerfeld\  radiation\  condition.}
\end{cases}
\end{align}
According to the expansion form (\ref{uFE}) and $\frac{1}{\mu_m}\frac{\partial u_\epsilon}{\partial \nu}\big{|}_{+}=\frac{1}{\mu_c}\frac{\partial u_\epsilon}{\partial \nu}\big{|}_{-}$ on $\partial D_\epsilon$, it follows that
\begin{align*}
\frac{1}{\mu_m}\frac{\partial u_0}{\partial\nu}\bigg{|}_{+}-\frac{1}{\mu_c}\frac{\partial u_0}{\partial\nu}\bigg{|}_{-}=0\hspace*{5.5cm} \mathrm{on}\ \partial D,
\end{align*}
\begin{align}
\frac{1}{\mu_m}\frac{\partial u_1}{\partial\nu}\bigg{|}_{+}-\frac{1}{\mu_c}\frac{\partial u_1}{\partial\nu}\bigg{|}_{-}&={h}\bigg{(}\frac{1}{\mu_c}\frac{\partial^2 u_{0}}{\partial\nu^2}\bigg{|}_{-}-\frac{1}{\mu_m}\frac{\partial^2 u_{0}}{\partial\nu^2}\bigg{|}_{+}\bigg{)}+\bigg{(}\frac{1}{\mu_c}\nabla u_{0}\cdot\nu_1\bigg{|}_{-}-\frac{1}{\mu_m}\nabla u_{0}\cdot\nu_1\bigg{|}_{+}\bigg{)}\nonumber\\
\label{uv+-}
&=\bigg{(} \frac{1}{\mu_m}-\frac{1}{\mu_c}\bigg{)}\frac{d}{dt}\bigg{(}h\frac{du_0}{dt}\bigg{)}+h\omega^2(\varepsilon_m-\varepsilon_c)u_0
\ \ \ \ \   \mathrm{on}\ \partial D.
\end{align}
The last equality in (\ref{uv+-}) is proved by using the representation of the Helmholtz operator in (\ref{Helm opra}). Similarly, since $u_\epsilon\mid_{+}=u_\epsilon\mid_{-}$ on $\partial D_\epsilon$, we have
\begin{align}
u_0|_{+}-u_0|_{-}&=0 \hspace*{7.4cm} \mathrm{on}\ \partial D,\nonumber\\
\label{u1+-}
u_1|_{+}-u_1|_{-}&=h\bigg{(}\frac{\partial u_0}{\partial\nu}\bigg{|}_{-}-\frac{\partial u_0}{\partial\nu}\bigg{|}_{+}\bigg{)}=h\bigg{(}1-\frac{\mu_m}{\mu_c}\bigg{)}\frac{\partial u_0}{\partial \nu}\bigg{|}_{-}
\hspace*{.6cm}\mathrm{on}\ \partial D.
\end{align}

Then the following asymptotic expansion results can be obtained. It is emphasized that above derivation is formal, and the asymptotic results will be proved strictly by layer potential techniques.
\begin{thm}\label{thm1}
The following asymptotic formula formally holds:
\begin{align*}
u_\epsilon(x)=u(x)+\epsilon u_1(x)+O(\epsilon^{2}),
\end{align*}
where the remainder $O(\epsilon^{2})$ depends on $N, D, k_c, k_m$, the $C_1$-norm of $h$, and $u_1$ is the unique solution to
\begin{align}
\label{u1equa}
\begin{cases}
&\displaystyle(\Delta +\omega^2\varepsilon_c\mu_c) u_1=0\hspace*{8.75cm}\mathrm{in}\   D,\medskip\\
&\displaystyle(\Delta +\omega^2\varepsilon_m\mu_m) u_1=0\hspace*{8.5cm} \mathrm{in}\  \mathbb{R}^2\backslash \overline{D},\medskip\\
&u_1|_{+}-u_1|_{-}=h\big{(}1-\frac{\mu_m}{\mu_c}\big{)}\frac{\partial u_0}{\partial\nu}|_{-}\  \hspace*{6.8cm} \mathrm{on}\ \partial D,\medskip\\
&\displaystyle\frac{1}{\mu_m}\frac{\partial u_1}{\partial\nu}\bigg{|}_{+}-\frac{1}{\mu_c}\frac{\partial u_1}{\partial\nu}\bigg{|}_{-}=
\bigg{(}\frac{1}{\mu_m}-\frac{1}{\mu_c}\bigg{)}\frac{d}{dt}\bigg{(}h\frac{du_0}{dt}\bigg{)}
+h\omega^2(\varepsilon_m-\varepsilon_c)u_0\ \hspace*{.7cm} \mathrm{on}\ \partial D,\medskip\\
&\displaystyle u_1\ \ \mathrm{satisfies\  the\  Sommerfeld\  radiation\  condition.}
\end{cases}
\end{align}
\end{thm}

\begin{proof}
According to (\ref{phipsin}) and the layer potential theory, one can show that
\begin{align}\label{u1}
u_1(x)=
\begin{cases}
&\displaystyle\mathcal{S}_D^{k_m}[\psi^{(1)}](x)-\mathcal{S}_D^{k_m}[\tau h\psi^{(0)}](x)+\mathcal{D}_D^{k_m}[h\psi^{(0)}](x),\hspace*{.70cm} x\in \mathbb{R}^2\setminus\overline{D},\medskip\\
&\displaystyle\mathcal{S}_D^{k_c}[\phi^{(1)}](x)-\mathcal{S}_D^{k_c}[\tau h\phi^{(0)}](x)+\mathcal{D}_D^{k_c}[h\phi^{(0)}](x),\hspace*{1.15cm} x\in {D},\\
\end{cases}
\end{align}
satisfies (\ref{un}), (\ref{uv+-}) and (\ref{u1+-}). In fact, by combining (\ref{Sdk1}) and the first equation of (\ref{phipsin}), we obtain
\begin{align*}
u_1|_{-}-u_1|_{+}=&h\frac{\partial u^i}{\partial\nu}+\mathcal{S}_{D,k_m}^{(1)}[\psi^{(0)}]-\mathcal{S}_{D,k_c}^{(1)}[\phi^{(0)}]
+\mathcal{S}_D^{k_m}[\tau h \psi^{(0)}]-\mathcal{S}_D^{k_c}[\tau h \phi^{(0)}]\\
&+(\mathcal{D}_D^{k_c}[h\phi^{(0)}])\bigg{|}_{-}-(\mathcal{D}_D^{k_m}[h\psi^{(0)}])\bigg{|}_{+}\\
=&h\bigg{(}\frac{\partial u^i}{\partial\nu}+\frac{\partial(\mathcal{S}^{k_m}_D[\psi^{(0)}])}{\partial\nu}\bigg{|}_{+}\bigg{)}
-h\frac{\partial(\mathcal{S}^{k_c}_D[\phi^{(0)}])}{\partial\nu}\bigg{|}_{-}\\
=&h\bigg{(}\frac{\mu_m}{\mu_c}-1\bigg{)}\frac{\partial u_0}{\partial \nu}\bigg{|}_{-} \ \ \ \ \ \ \mathrm{on}\ \partial D,
\end{align*}
Since the second equation of (\ref{phipsin}), we have
\begin{align*}
\frac{1}{\mu_c}\frac{\partial u_1}{\partial\nu}\bigg{|}_{-}-\frac{1}{\mu_m}\frac{\partial u_1}{\partial \nu}\bigg{|}_{+}=&\frac{1}{\mu_m}(\mathcal{K}_{D,k_m}^{(1)}[\psi^{(0)}])\bigg{|}_{+}-
\frac{1}{\mu_c}(\mathcal{K}_{D,k_c}^{(1)}[\phi^{(0)}])\bigg{|}_{-}+\frac{1}{\mu_m}\bigg{(}h\frac{\partial^2 u^i}{\partial \nu^2}-h'\frac{\partial u^i}{\partial T}\bigg{)}\\
&+\frac{1}{\mu_m}\frac{\partial (\mathcal{S}_D^{k_m}[\tau h\psi^{(0)}])}{\partial \nu}\bigg{|}_{+}-\frac{1}{\mu_c}\frac{\partial (\mathcal{S}_D^{k_c}[\tau h\phi^{(0)}])}{\partial \nu}\bigg{|}_{-}\\
&+\frac{1}{\mu_c}\frac{\partial (\mathcal{D}_D^{k_c}[h\phi^{(0)}])}{\partial \nu}\bigg{|}_{-}-\frac{1}{\mu_m}\frac{\partial (\mathcal{D}_D^{k_m}[ h\psi^{(0)}])}{\partial \nu}\bigg{|}_{+}\ \ \ \ \ \ \  \mathrm{on}\ \partial D,
\end{align*}
owing to $(\Delta+k_m^2)u^i=0$ in $\mathbb{R}^2$, and combined with \eqref{Helm opra}, we find
\begin{align*}
h\frac{\partial^2 u^i}{\partial \nu^2}-h'\frac{\partial u^i}{\partial T}=
\tau h\frac{\partial u^i}{\partial\nu}-\frac{d}{dt}\bigg{(}h\frac{d u^i}{dt}\bigg{)}-hk_m^2u^i\ \ \ \mathrm{on}\ \partial D.
\end{align*}
It is evident from (\ref{k1}) that $u_1$ satisfies (\ref{uv+-}). This completes the proof.
\end{proof}

We introduce the following definition of a shape sensitivity functional, to study the variation in measurements caused by changes in the shape of the inclusion $D$.
\begin{defn}\label{de1}
The shape sensitivity functional for the measurement $u^s|_{\partial \Omega}$ with respect to the shape of $\partial D$ is defined as
\begin{align*}
SSF(\partial D):=\lim_{\epsilon\rightarrow0}\frac{u_\epsilon^s|_{\partial \Omega}-u^s|_{\partial \Omega}}{\epsilon}.
\end{align*}
\end{defn}

\begin{rem}
From Definition \ref{de1}, the shape sensitivity functional is actually the shape derivative of the forward operator $\mathcal{F}(\partial D)$ (cf. \cite{Lim2012}). Furthermore, by using Theorem \ref{thm1} and \eqref{u1}, the shape sensitivity functional can be rewritten as
\begin{align}
\label{ssf1}
SSF(\partial D)(x)=(\mathcal{S}_D^{k_m}[\psi^{(1)}]
-\mathcal{S}_{D}^{k_m}[\tau h\psi^0]+\mathcal{D}_{D}^{k_m}[h \psi^{0}])(x), \ \ x\in \partial \Omega.
\end{align}
\end{rem}

\subsection{Spectral representation of the shape sensitivity functional}\label{sect:3.2}

In this subsection, we derive the spectral representation of the shape sensitivity functional. It indicates that, when plasmon resonances occur, the shape sensitivity functional is amplified and exhibits a prominent peak. Hence, the plasmon resonance can significantly increase the sensitivity of the near-field measurement with respect to the shape of a domain. For simplicity, in the subsequent spectral analysis, we always exclude the essential spectrum 0 from the spectrum set of the NP operator.

Let $d=dist(\lambda, \sigma(\mathcal{K}_D^*))$ be the distance of $\lambda$ to the spectrum set $\sigma(\mathcal{K}_D^*)$. Next, we are ready to state our main result in the following theorem.

\begin{thm}\label{thm3.4}
 In the quasi-static regime, as $d\rightarrow 0$, if $\omega^2\log \omega=o(d)$,
 then the shape sensitivity functional has the following spectral expansion:
\begin{align}\label{ssf}
SSF(\partial D)=&\frac{\mu_c}{\mu_m-\mu_c}\sum_{q=1}^3\frac{(-p_c)^{q-1}}{(\log \omega)^q}\sum_{l\in J}\frac{(E_2,\varphi_l)_{\mathcal{H}^*}\widetilde{\mathcal{S}}_D[\varphi_l]}{\lambda-\lambda_l}
+\sum_{l\in J}\sum_{j\in J}\frac{\mathrm{i}k_m(d\cdot\nu,\varphi_j)_{\mathcal{H}^*}
(E_1,\varphi_m)_{\mathcal{H}^*}\widetilde{\mathcal{S}}_D[\varphi_m]}{(\lambda-\lambda_j)(\lambda-\lambda_l)}
\nonumber\\
&+\sum_{j\in J}\frac{\mathrm{i}k_m(d\cdot \nu, \varphi_j)_{\mathcal{H}^*}(-\widetilde{\mathcal{S}}_D[\tau h\varphi_j]+\mathcal{D}_D[h\varphi_j])}{\lambda-\lambda_j}
+O\bigg{(}\frac{1}{d(\log \omega)^3}\bigg{)}+O\bigg{(}\frac{\omega^2(\log \omega)^2}{d^2}\bigg{)}\nonumber\\
&+O\bigg{(}\frac{\omega^3(\log \omega)^3}{d^3}\bigg{)}+O\big{(}\omega\log \omega\big{)},
\end{align}
where
\begin{align*}
E_1&=\tau h\lambda_j\varphi_j-\mathcal{K}_D^*[\tau h\varphi_j]+\mathcal{D}_D[h\varphi_j]-\frac{d}{dt}\bigg{(}h\frac{d(\widetilde{\mathcal{S}}_D\varphi_j)}{dt}\bigg{)},\\ E_2&=2\pi\bigg{(}\frac{\partial D_D[h\varphi_0]}{d \nu}+\bigg{(}\frac{1}{2}Id-\mathcal{K}_D^*\bigg{)}(\widetilde{\mathcal{S}}_D^{-1}[h\varphi_0/2]
+\widetilde{\mathcal{S}}_D^{-1}\mathcal{K}_D[h\varphi_0])\bigg{)}
\end{align*}
with the constant $p_c=2\pi\mathcal{S}_D[\varphi_0]+\log(\sqrt{\varepsilon_c\mu_c})+\gamma-\log2-\mathrm{i}\pi/2$.
\end{thm}

\begin{proof}
We reformulate the form of $\Upsilon_{k_m}$ in (\ref{Sk}) as follows
\begin{align*}
\Upsilon_{k_m}[\cdot]=\frac{1}{2\pi}\log\omega(\cdot,\varphi_0)_{\mathcal{H}^*}+A_m(\cdot,\varphi_0)_{\mathcal{H}^*},
\end{align*}
where the constant $A_m=\mathcal{S}_D[\varphi_0]+\chi(\partial D)+\frac{1}{2\pi}(\log (\sqrt{\varepsilon_m\mu_m})+\gamma-\log 2)-\frac{\mathrm{i}}{4}$. Similarly, we give the form of $\Upsilon_{k_c}[\cdot]$, where $A_c$ pairs are used for the constants in $\Upsilon_{k_c}[\cdot]$. The operator $\mathcal{U}_{k_c}$ in \eqref{invS} also have the following asymptotic formula:
\begin{align*}
\mathcal{U}_{k_c}=-\bigg{(}\frac{2\pi(\widetilde{\mathcal{S}}_D^{-1}[\cdot],\varphi_0)_{\mathcal{H}^*}\varphi_0}
{\log \omega}-\frac{2\pi p_c(\widetilde{\mathcal{S}}_D^{-1}[\cdot],\varphi_0)_{\mathcal{H}^*}\varphi_0} {(\log\omega)^2}+\frac{2\pi p_c^2(\widetilde{\mathcal{S}}_D^{-1}[\cdot],\varphi_0)_{\mathcal{H}^*}\varphi_0} {(\log\omega)^3}+O\bigg{(}\frac{1}{(\log\omega)^4}\bigg{)}\bigg{)}.
\end{align*}
Since $\omega^2\log\omega=o(d)$, the expansion of $\psi^{(0)}$ in (\ref{den}) can be obtained
\begin{align}\label{psik}
\psi^{(0)}=\sum_{j\in J}\frac{\mathrm{i}k_m(d\cdot\nu,\varphi_j)_{\mathcal{H}^*}\varphi_j}{\lambda-\lambda_j}+O(\omega)
+O\bigg{(}\frac{\omega^2}{d}\bigg{)}+O\bigg{(}\frac{\omega^3\log \omega}{d^2}\bigg{)}
:=\psi_{1}^{(0)}+\mathcal{R},
\end{align}
where $\mathcal{R}$ represents a higher order term.

As $\omega\rightarrow 0$, combined with (\ref{phipsin}), we have $\phi^{(1)}=(\mathcal{S}_D^{k_c})^{-1}(\mathcal{S}_D^{k_m}[\psi^{(1)}]-Q_1)$, and $Q_1=-h\frac{\partial u^i}{\partial \nu}+\mathcal{S}_{D,k_c}^{(1)}[\phi^{(0)}]
-\mathcal{S}_{D,k_m}^{(1)}[\psi^{(0)}]$, whereas the following equation hold for $\psi^{(1)}$
\begin{align}\label{psi1}
\mathcal{A}_D(\omega)[\psi^{(1)}]=g,
\end{align}
where $g=Q_2+\frac{1}{\mu_c}(\frac{1}{2}Id-(\mathcal{K}_D^{k_c})^*)(\mathcal{S}_D^{k_c})^{-1}[Q_1]$, $Q_2=-\frac{1}{\mu_m}(h\frac{\partial^2 u^i}{\partial\nu^2}-h'\frac{\partial u^i}{\partial T})+\frac{1}{\mu_c}\mathcal{K}_{D,k_c}^{(1)}[\phi^{(0)}]-\frac{1}{\mu_m}\mathcal{K}_{D,k_m}^{(1)}[\psi^{(0)}]$.

 In order to calculate $\psi^{(1)}$ of the (\ref{psi1}), we first calculate  $\phi^{(0)}=(\mathcal{S}_D^{k_c})^{-1}(\mathcal{S}_D^{k_m}[\psi^{(0)}]+u^i)$ of the $Q_1$.
 According to (\ref{psik}), (\ref{invS}) and the fact $\mathcal{L}_D(\chi(\partial D))=0$, $(\psi^{(0)}_1,\varphi_0)_{\mathcal{H}^*}=0$, by straightforward calculation, we have
\begin{align}\label{phi0}
\phi^{(0)}&=\psi^{(0)}_1+\sum_{q=1}^3\frac{2\pi\varphi_0(-p_c)^{q-1}}{(\log\omega)^q}+O\bigg{(}\frac{1}{(\log \omega)^4}\bigg{)}+\mathcal{R}.
\end{align}
Next, we can calculate the $\frac{1}{\mu_c}\big{(}\frac{1}{2}Id-(\mathcal{K}_D^{k_c})^*\big{)}(\mathcal{S}_D^{k_c})^{-1}[Q_1]$
of $g$. Since
\begin{align*}
\bigg{(}\frac{1}{2}Id-\mathcal{K}_D^*\bigg{)}\mathcal{U}_{k_c}=0,
\end{align*}
and according to (\ref{K^*}) and (\ref{invS}), it is evident that
\begin{align}\label{in22}
\frac{1}{\mu_c}\bigg{(}\frac{1}{2}Id-(\mathcal{K}_D^{k_c})^*\bigg{)}(\mathcal{S}_D^{k_c})^{-1}
=\frac{1}{\mu_c}\bigg{(}\frac{1}{2}Id-\mathcal{K}_D^*\bigg{)}\widetilde{\mathcal{S}}_D^{-1}+O(\omega^2\log \omega).
\end{align}
Combining (\ref{invS}), (\ref{Sdk1}), (\ref{in22}) and using the fact  $\mathcal{L}_D\Upsilon_{k_c}=0$, we deduce
\allowdisplaybreaks
\begin{align}\label{ll}
&\frac{1}{\mu_c}\bigg{(}\frac{1}{2}Id-(\mathcal{K}_D^{k_c})^*\bigg{)}(\mathcal{S}_D^{k_c})^{-1}[Q_1]\nonumber\\
=&\frac{1}{\mu_c}\bigg{(}\frac{1}{2}Id-\mathcal{K}_D^{k_c})^*\bigg{)}(\mathcal{S}_D^{k_c})^{-1}\bigg{(}-h\frac{\partial u^i}{\partial \nu}-\mathcal{S}_{D}^{k_c}[\tau h\phi^{(0)}]+h(\mathcal{K}^{k_c}_D)^*[\phi^{(0)}]+\mathcal{K}^{k_c}_D[h\phi^{(0)}]\nonumber\\
&+\mathcal{S}_{D}^{k_m}[\tau h\psi^{(0)}]
-h(\mathcal{K}^{k_m}_D)^*[\psi^{(0)}]-\mathcal{K}^{k_c}_D[h\psi^{(0)}]\bigg{)}\nonumber\\
=&\frac{1}{\mu_c}\bigg{(}\frac{1}{2}Id-\mathcal{K}_D^*\bigg{)}\widetilde{\mathcal{S}}_D^{-1}\bigg{(}-
(\widetilde{\mathcal{S}}_D+\Upsilon_{k_c})\bigg{(}\tau h\psi^{(0)}_1+\sum_{q=1}^3\frac{2\pi\tau h\varphi_0(-p_c)^{q-1}}
{(\log\omega)^q}\bigg{)}
+h\mathcal{K}_D^*[\psi^{(0)}_1]\nonumber\\
&+\sum_{q=1}^3\frac{2\pi (-p_c)^{q-1}h\mathcal{K}_D^*[\varphi_0]}{(\log\omega)^q}
+\mathcal{K}_D[h\psi^{(0)}_1]+\sum_{q=1}^3\frac{2\pi(-p_c)^{q-1}\mathcal{K}_D[h\varphi_0]}{(\log\omega)^q}
+(\widetilde{\mathcal{S}}_D+\Upsilon_{k_m})[\tau h\psi^{(0)}_1]\nonumber\\
&-h\mathcal{K}_D^*[\psi^{(0)}_1]-\mathcal{K}_D[h\psi^{(0)}_1]\bigg{)}
+O\bigg{(}\frac{1}{(\log \omega)^4}\bigg{)}+\mathcal{R}\nonumber\\
=&\frac{2\pi}{\mu_c}\bigg{(}\frac{1}{2}Id-\mathcal{K}_D^*\bigg{)}\bigg{(}\sum_{q=1}^3\frac{(-p_c)^{q-1}\tau h\varphi_0+\widetilde{\mathcal{S}}_D^{-1}[h\varphi_0/2]+\widetilde{\mathcal{S}}_D^{-1}\mathcal{K}_D[h\varphi_0])}
{(\log\omega)^q}\bigg{)}+O\bigg{(}\frac{1}{(\log \omega)^4}\bigg{)}+\mathcal{R}.
\end{align}

Then we need to calculate the asymptotic expansion of $Q_2$. First for item $\frac{d}{dt}\big{(}h\frac{d(\mathcal{S}_D^{k_c}[\phi^{(0)}])}{dt}\big{)}$, from (\ref{phi0}) and (\ref{series}), we deduce
\allowdisplaybreaks
\begin{align}\label{ds}
\frac{d}{dt}\bigg{(}h\frac{d(\mathcal{S}_D^{k_c}[\phi^{(0)}])}{dt}\bigg{)}
=\frac{d}{dt}\bigg{(}h\frac{d(\widetilde{\mathcal{S}}_D[\psi_1^{(0)}])}{dt}\bigg{)}+\mathcal{R}.
\end{align}
Similarly, we compute
\begin{align}\label{dsm}
\frac{d}{dt}\bigg{(}h\frac{d(\mathcal{S}_D^{k_m}[\psi^{(0)}])}{dt}\bigg{)}
=\frac{d}{dt}\bigg{(}h\frac{d(\widetilde{\mathcal{S}}_D[\psi^{(0)}_1])}{dt}\bigg{)}+\mathcal{R}\log\omega,
\end{align}
combining formula (\ref{k1}), (\ref{ds}) and (\ref{dsm}), we ascertain
\allowdisplaybreaks
\begin{align}\label{Q2}
Q_2=&-\frac{1}{\mu_m}\bigg{(}h\frac{\partial^2 u^i}{\partial \nu^2}-h'\frac{\partial u^i}{\partial T}\bigg{)}+\frac{1}{\mu_c}\mathcal{K}_{D,k_c}^{(1)}[\phi^{(0)}]-\frac{1}{\mu_m}\mathcal{K}_{D,k_m}^{(1)}[\psi^{(0)}]\nonumber\\
=&-\frac{1}{\mu_m}\bigg{(}h\frac{\partial^2 u^i}{\partial \nu^2}-h'\frac{\partial u^i}{\partial T}\bigg{)}+\frac{1}{\mu_c}\bigg{(}\tau h(\mathcal{K}_D^{k_c})^*[\phi^{(0)}]-(\mathcal{K}_D^{k_c})^*[\tau h\phi^{(0)}]+\frac{\partial (\mathcal{D}_D^{k_c}[h\phi^{(0)}])}{\partial \nu}\nonumber\\
&-\frac{d}{dt}\bigg{(}h\frac{d(\mathcal{S}_D^{k_c}[\phi^{(0)}])}{dt}\bigg{)}-k_c^2h\mathcal{S}_D^{k_c}[\phi^{(0)}]\bigg{)}
-\frac{1}{\mu_m}\bigg{(}\tau h(\mathcal{K}_D^{k_c})^*[\psi^{(0)}]-(\mathcal{K}_D^{k_c})^*[\tau h\psi^{(0)}]\nonumber\\
&+\frac{\partial (\mathcal{D}_D^{k_c}[h\psi^{(0)}])}{\partial \nu}
-\frac{d}{dt}\bigg{(}h\frac{d(\mathcal{S}_D^{k_c}[\psi^{(0)}])}{dt}\bigg{)}
-k_m^2h\mathcal{S}_D^{k_c}[\psi^{(0)}]\bigg{)}\nonumber\\
=&\frac{(\mu_m-\mu_c)}{\sqrt{\mu_m}\mu_c}\sum_{j\in J}\frac{\mathrm{i}k_m(d\cdot\nu,\varphi_j)_{\mathcal{H}^*}(\tau h\lambda_j\varphi_j-\mathcal{K}_D^*[\tau h\varphi_j]+\mathcal{D}_D[h\varphi_j]-\frac{d}{dt}\big{(}
h\frac{d(\widetilde{\mathcal{S}}_D[\varphi_j])}{dt}\big{)})}{\lambda-\lambda_j}\nonumber\\
&+\frac{2\pi}{\mu_c}\sum_{q=1}^3\frac{(-p_c)^{q-1}}{(\log \omega)^q}\bigg{(}\tau h\varphi_0/2-\mathcal{K}_D^*[\tau h\varphi_0]+\frac{\partial (\mathcal{D}_D [h\varphi_0])}{\partial \nu}\bigg{)}+O\bigg{(}\frac{1}{(\log \omega)^4}\bigg{)}\\
&+\mathcal{R}+\log\omega \mathcal{R}.\nonumber
\end{align}
Furthermore, owing to (\ref{ll}), (\ref{Q2}), we find
\begin{align}\label{f}
f=&Q_2+\frac{1}{\mu_c}\bigg{(}\frac{1}{2}Id-(\mathcal{K}_D^{k_c})^*\bigg{)}(\mathcal{S}_D^{k_c})^{-1}[Q_1]\nonumber\\
=&\frac{(\mu_m-\mu_c)}{\sqrt{\mu_m}\mu_c}\sum_{j\in J}\frac{\mathrm{i}\omega\sqrt{\varepsilon_m}(d\cdot\nu,\varphi_j)_{\mathcal{H}^*}(\tau h\lambda_j\varphi_j-\mathcal{K}_D^*[\tau h\varphi_j]+\mathcal{D}_D[h\varphi_j]-\frac{d}{dt}\big{(}h\frac{d(\widetilde{\mathcal{S}}_D[\varphi_j])}{dt}\big{)})}{\lambda-\lambda_j}\nonumber\\
&+\frac{2\pi}{\mu_c}\bigg{(}\frac{\partial( D_D[h\varphi_0])}{\partial \nu}+\sum_{q=1}^3\frac{(-p_c)^{q-1}}{(\log \omega)^q}\bigg{(}\frac{1}{2}Id-\mathcal{K}_D^*\bigg{)}(\widetilde{\mathcal{S}}_D^{-1}[h\varphi_0/2]
+\widetilde{\mathcal{S}}_D^{-1}\mathcal{K}_D[h\varphi_0])\nonumber\\
&+O\bigg{(}\frac{1}{(\log\omega)^4}\bigg{)}
+O\bigg{(}\frac{\omega\log\omega}{d}\bigg{)}+O\bigg{(}\frac{\omega^3(\log\omega)^2}{d^2}\bigg{)}\nonumber\\
:=&\frac{(\mu_m-\mu_c)}{\sqrt{\mu_m}\mu_c}\sum_{j\in J}\frac{\mathrm{i}\omega\sqrt{\varepsilon_m}(d\cdot\nu,\varphi_j)_{\mathcal{H}^*}E_1}{\lambda-\lambda_j}
+\frac{E_2}{\mu_c}\bigg{(}\sum_{q=1}^3\frac{(-p_c)^{q-1}}{(\log \omega)^q}\bigg{)}
+O\bigg{(}\frac{1}{(\log\omega)^4}\bigg{)}\\
&+O\bigg{(}\frac{\omega\log\omega}{d}\bigg{)}+O\bigg{(}\frac{\omega^3(\log\omega)^2}{d^2}\bigg{)}.\nonumber
\end{align}
Recalling (\ref{psi1}), and combined with (\ref{f}), (\ref{per lam}), we deduce
\begin{align}\label{psi(1)}
\psi^{(1)}=&\mathcal{A}_D^{-1}(\omega)[f]\nonumber\\
=&\sum_{j\in J}\frac{(f,\widetilde{\varphi}_j)_{\mathcal{H}^*}\widetilde{\varphi}_j}{\tau_j(\omega)}
+\mathcal{A}_D^{-1}(\omega)[P_{J^c}[f]]\nonumber\\
=&\frac{\mu_m}{\mu_m-\mu_c}\bigg{(}\sum_{q=1}^3\frac{(-p_c)^{q-1}}{(\log \omega)^q}\sum_{l\in J}\frac{(E_2,\varphi_l)_{\mathcal{H}^*}\varphi_l}{\lambda-\lambda_l}
\bigg{)}+\sum_{l\in J}\sum_{j\in J}\frac{\mathrm{i}k_m(d\cdot\nu,\varphi_j)_{\mathcal{H}^*}(E_1,\varphi_l)_{\mathcal{H}^*}\varphi_l}
{(\lambda-\lambda_j)(\lambda-\lambda_l)}\\
&+O\bigg{(}\frac{1}{d(\log \omega)^4 }\bigg{)}+O\bigg{(}\frac{\omega^2\log \omega}{d^2}\bigg{)}+O\bigg{(}\frac{\omega^3(\log \omega)^2}{d^3}\bigg{)}+O(\omega)\nonumber,
\end{align}
where the norm $\parallel \mathcal{A}_D^{-1}(\omega)[P_{J^c}(\omega)[f]]\parallel_{\mathcal{L}(\mathcal{H}^*(\partial D), \mathcal{H}^*(\partial D))}$ is uniformly bounded in $\omega$ (see \cite{Millien2017}). Here, $P_{J}$ is a projection operator with $P_{J}[\varphi_j]=\varphi_j$ for $j\in J$, and $P_{J^c}[\varphi_j]=0$ for $j\in J^c$. Thus we have $\mathcal{A}_D^{-1}(\omega)[P_{J^c}(\omega)[f]]=O(\omega)$.

Since $\mathcal{S}_D^{k_m}=\widetilde{\mathcal{S}}_D+\Upsilon_{k_m}+O(\omega^2\log \omega)$, $\mathcal{D}_D^{k_m}=\mathcal{D}_D+O(\omega^2\log \omega)$ and combining (\ref{psi(1)}), (\ref{psik}) and (\ref{ssf1}), a direct calculation shows that the theorem immediately follows.
\end{proof}

\begin{rem}
From Theorem \ref{thm3.4}, it is clear that the $l$-th mode in the expansion formula contains both $O\left(\frac{1}{\lambda-\lambda_l}\right)$, $O\left(\frac{1}{(\lambda-\lambda_l)(\lambda-\lambda_j)}\right)$ and
$O\left(\frac{1}{(\lambda-\lambda_l)^2}\right)$.  Therefore, for the sufficiently small loss ($\Im(\mu_c)\rightarrow0$), the $l$-th mode will exhibit a large peak if the plasmon resonance condition (\ref{plasmon}) is satisfied.
\end{rem}

\begin{rem}
Since the logarithmic singularity $\log|x-y|$ (as $x\rightarrow y$) and $\log k$ (as $k\rightarrow0$) in 2-D fundamental solution of Helmholtz equation, the logarithmic singularity are also inherited to the spectral expansion of shape sensitivity functional (\ref{ssf}). It is very different from the 3-D case and static electric field case \cite{DLZ2022}.
\end{rem}

According to Theorem \ref{thm3.4}, the upper bound estimate of the sensitivity functional can be derived easily, and we state the following corollary.
\begin{cor}\label{cor:1}
If the assumptions of Theorem \ref{thm3.4} are satisfied, and the mapping $SSF:H^1(\partial D)\rightarrow H^2(\partial \Omega)$, then the shape sensitivity functional satisfies
\begin{align*}
\| SSF(\partial D)\|_{H^2(\partial \Omega)}
&\leq C\bigg{(}\sum_{q=1}^3\frac{1}{(\log \omega)^q}\sum_{l\in J}\frac{1}{\lambda-\lambda_l}
+k_m\sum_{l\in J}\sum_{j\in J}\frac{|\partial D|^{\frac{1}{2}}(\|\tau\|_{C(\partial D)}+1)}{(\lambda-\lambda_l)(\lambda-\lambda_j)})\\
&+k_m\sum_{j\in J}\frac{|\partial D|^{\frac{1}{2}}(\|\tau\|_{C(\partial D)}+1)}{\lambda-\lambda_j}\bigg{)}+O\bigg{(}\frac{1}{d(\log \omega)^3}\bigg{)}+O\bigg{(}\frac{\omega^2(\log \omega)^2}{d^2}\bigg{)}\\
&+O\bigg{(}\frac{\omega^3(\log \omega)^3}{d^3}\bigg{)}+O{(}\omega\log \omega{)}.
\end{align*}
Moreover, when $d\cdot \nu=0$, then the following result holds
\begin{align*}
\parallel SSF(\partial D)\parallel_{H^2(\partial \Omega)}
&\leq C\bigg{(}\sum_{q=1}^3\frac{1}{(\log \omega)^q}\sum_{l\in J}\frac{1}{\lambda-\lambda_l}\bigg{)}+O\bigg{(}\frac{1}{d(\log \omega)^3}\bigg{)}+O\bigg{(}\frac{\omega^2(\log \omega)^2}{d^2}\bigg{)}\\
&+O\bigg{(}\frac{\omega^3(\log \omega)^3}{d^3}\bigg{)}+O(\omega\log \omega),
\end{align*}
where $C$ is a constant.
\end{cor}

\begin{rem}
From Corollary \ref{cor:1}, we can see that the $H^2$-norm of shape sensitivity functional is controlled by the curvature and size of $\partial D$. Especially if the incident direction and normal vector are perpendiculars, i.e., $d\cdot \nu=0$, the upper bound of shape sensitivity functional is independent of the curvature and size of $\partial D$.
\end{rem}

\subsection{ Spectral expansion of a circle}
For the special case when the domain $D$  is a disk, we have that $\sigma(\mathcal{K}^*_D)=\{ 0,\frac{1}{2}\}$. In the previous analysis, we remove $0$ and $\frac{1}{2}$ in our hypothesis. In this subsection, we focus  on the spectral expansion of the shape sensitivity functional when $D$  is a disk. Assuming that $D$ is a disk of radius $r_0$ located at the origin, the perturbation boundary $\partial D_\epsilon$ is given by
\begin{align*}
\partial D_\epsilon=(r_0+\varepsilon h(t))\bigg{(} \begin{array}{c}
\cos t\\
\sin t
\end{array} \bigg{)},\ t\in[0,2\pi]   \}.
\end{align*}

From the Theorem \ref{thm1}, we have $u_\epsilon=u_0(x)+\epsilon u_1(x)+O(\epsilon^2)$, where $u_0$, $u_1$ fulfill (\ref{sca equ}) and (\ref{u1}),  respectively. We can obtain the explicit solution of (\ref{sca equ}) and (\ref{u1}) as follows
\begin{align*}
u_0(r,t)=
\begin{cases}
\sum\limits_{n\in \mathds{Z}}e^{\mathrm{i}n(\frac{\pi}{2}-t_d)}J_n(k_m r)e^{\mathrm{i}nt}+a_n H^{(1)}_n(k_m r)e^{\mathrm{i}nt}, \ r>r_0,\medskip\\
\sum\limits_{n\in \mathds{Z}}b_nJ_n(k_c r)e^{\mathrm{i}nt}, \ r\leq r_0,
\end{cases}
\end{align*}
and
\begin{align*}
u_1(r,t)=
\begin{cases}
\sum\limits_{n\in \mathds{Z}}c_n H^{(1)}_n(k_m r)e^{\mathrm{i}nt}, \ r>r_0,\medskip\\
\sum\limits_{n\in \mathds{Z}}d_nJ_n(k_c r)e^{\mathrm{i}nt}, \ r\leq r_0,
\end{cases}
\end{align*}
where $J_n, H^{(1)}_n$ are the Bessel function of order $n\in N$, the incidence direction $d=(\cos(t_d), \sin(t_d))$ and the plane wave has the following expansion
\begin{align*}
e^{\mathrm{i}kx\cdot d}=\sum_{n\in \mathbb{Z}}e^{\mathrm{i}n(\frac{\pi}{2}-t_d)}J_n(kr)e^{\mathrm{i}nt}.
\end{align*}
Then the constant $a_n$ and $b_n$ are given by
\begin{align}\label{an}
a_n&=\frac{e^{\mathrm{i}n(\frac{\pi}{2}-t_d)}\bigg{(}\frac{k_m}{\mu_m}J'_n(k_mr_0)J_n(k_c r_0)-\frac{k_c}{\mu_c}J'_n(k_cr_0)J_n(k_m r_0)\bigg{)}}{\frac{k_c}{\mu_c}J'_n(k_c r_0)H_n^{(1)}(k_m r_0)-\frac{k_m}{\mu_m}J_n(k_c r_0)H_n^{(1)'}(k_m r_0)},\medskip\\
\label{bn}
b_n&=\frac{e^{\mathrm{i}n(\frac{\pi}{2}-t_d)}\frac{k_m}{\mu_m}\bigg{(}J'_n(k_mr_0)H_n^{(1)}(k_m r_0)-J_n(k_mr_0)H_n^{(1)'}(k_m r_0)\bigg{)}}{\frac{k_c}{\mu_c}J'_n(k_c r_0)H_n^{(1)}(k_m r_0)-\frac{k_m}{\mu_m}J_n(k_c r_0)H_n^{(1)'}(k_m r_0)}.
\end{align}
Owing to the (\ref{an}), (\ref{bn}) and (\ref{u1}), a direct calculation show that
\begin{align*}
c_n&=\frac{e^{\mathrm{i}n(\frac{\pi}{2}-t_d)}\frac{k_m}{\mu_m}\bigg{(}J'_n(k_mr_0)H_n^{(1)}(k_m r_0)-J_n(k_mr_0)H_n^{(1)'}(k_m r_0)\bigg{)}}{\frac{k_c}{\mu_c}J'_n(k_c r_0)H_n^{(1)}(k_m r_0)-\frac{k_m}{\mu_m}J_n(k_c r_0)H_n^{(1)'}(k_m r_0)}\times\medskip\\
&\frac{\bigg{(}\frac{1}{\mu_c}\big{(}1-\frac{\mu_m}{\mu_c}\big{)}h(t)k_c^2[J'_n(k_cr_0)]^2-
(n(\mathrm{i}h'(t)-nh)\big{(}\frac{1}{\mu_m}-\frac{1}{\mu_c}\big{)}-h(t)\omega^2(\varepsilon_m-\varepsilon_c))
[J_n(k_mr_0)]^2\bigg{)}}{\frac{k_c}{\mu_c}J'_n(k_c r_0)H_n^{(1)}(k_m r_0)-\frac{k_m}{\mu_m}J_n(k_c r_0)H_n^{(1)'}(k_m r_0)},
\end{align*}
as $\omega\rightarrow 0$, according to the asymptotic expansion in \cite{SKS2000} with respect to $J_n(z), H_n^{(1)}(z)$ and its derivative, $c_n$ can be obtained by a direct calculation,
\begin{align*}
c_0&=\frac{\frac{2}{\mathrm{i}\pi r_0 \mu_m}\cdot O(\omega^2)}{-(\frac{2}{\pi r_0 \mu_m})^2+O(\omega^2\log \omega)}=O(\omega^2),\medskip \\
c_n&=e^{\mathrm{i}n(\frac{\pi}{2}-t_d)}\bigg{[}\omega^{2n}
\frac{(\frac{\mathrm{i}\pi\varepsilon_m^n\mu_m^{n+1}\mu_c}
{2^{2n+1}(\mu_c-\mu_m)})\bigg{(}\frac{1}{\mu_c}h(t)\frac{1}{[(n-1)!]^2}r_0^{2n-1}-
\frac{1}{\mu_m}n(\mathrm{i}h'-nh)r_0^{2n+1}\bigg{)}}{\lambda^2}\nonumber\\
&+\omega^{2(n+1)}\frac{\frac{\mathrm{i}\pi\varepsilon_m^n
\mu_m^{n+1}\mu_c^2(\varepsilon_m-\varepsilon_c)}{2^{2n+1}(\mu_m-\mu_c)^2}h(t) r_0^{2n+1})}{\lambda^2}\bigg{]}.
\end{align*}

Finally, we are ready to state our main result in the subsection.
\begin{thm}
Let domain $D$ be a disk, then the shape sensitivity functional $SSF(\partial D)$ has the following asymptotic expansion in the quasi-static regime:
\begin{align*}
 SSF(\partial D)&=\sum_{n\geq 1}\frac{e^{\mathrm{i}n(\pi/2-t_d+t)}(n-1)!\varepsilon_m^{\frac{n}{2}}
\mu_m^{\frac{n}{2}+1}\mu_c}{4^n(\mu_c-\mu_m)R_0^n}
\bigg{(}\omega^{n}\frac{\frac{h r_0^{2n-1}}{\mu_c[(n-1)!]^2}-\frac{1}{\mu_m}n(\mathrm{i}h'-nh)r_0^{2n+1}}{\lambda^2}\nonumber\\
&+\omega^{n+1}\frac{\frac{\mu_c(\varepsilon_m-\varepsilon_c)}{\mu_m-\mu_c}hr_0^{2n+1}}{\lambda^2}\bigg{)} +O(\omega^2(\log \omega)).
\end{align*}
Moreover, the following estimation holds
\begin{align}\label{diskssf}
\parallel SSF(\partial D)\parallel_{C(\partial \Omega)}\leq\sum_{n\geq1}\bigg{(}\omega^n\frac{\widetilde{C}_n(r_0^{2n-1}+r_0^{2n+1})}{\lambda^2}
+\omega^{n+1}\frac{\widetilde{D}_nr_0^{2n+1}}{\lambda^2}\bigg{)}+O(\omega^2(\log \omega)),
\end{align}
where $\widetilde{C}_n, \widetilde{D}_n$ depend on $n,\varepsilon_c, \mu_c, \varepsilon_m, \mu_m, R_0, \parallel h\parallel_{C^1}$.
\end{thm}

\begin{rem}
Notice that, from Condition \ref{1}, the essential spectrum of Neumann-Poincar\'e operator for general smooth domain $0$ is excluded. However, when $D$ is a disk, the leading term of spectral expansion of SSF is $\omega/\lambda^2$, which is quite distinct from the general case. When $D$ is a general domain (non-disk), the leader term is $\frac{1}{(\log\omega)\cdot(\lambda-\lambda_l)}$ (see Theorem \ref{thm3.4}).
\end{rem}

\section{ Hierarchical Bayesian and Laplace approximation}\label{sect:4}

In this section, we discuss the numerical issues for shape reconstruction. First, we suppose that $\partial D$ is a starlike boundary curve for the origin, i.e. there exists $q\in C^2[0,2\pi]$ such that
\begin{align*}
\partial D=\{\vec{q}(t)=q(t)\bigg{(} \begin{array}{c}
\cos t\\
\sin t
\end{array} \bigg{)},\ t\in[0,2\pi]   \}.
\end{align*}
Note that without change of notation we rewrite $\mathcal{F}(q)=u^s|_{\partial\Omega}$ of (\ref{eq:ip1}).

One way to overcome the ill-posedness of the inverse problem is to use the Tikhonov regularization method, which proposes to augment the cost functional with a quadratic term, i.e.,
\begin{align}
\label{variational reg}
J[q]:=\frac{1}{2}\parallel \mathcal{F}(q)-u^{s,\delta}\parallel_2^2+\frac{\mu}{2}\parallel q \parallel_{L^2[0,2\pi]}^2,
\end{align}
where $u^{s,\delta}=(u_1^{s,\delta},u_2^{s,\delta},\cdots,u_n^{s,\delta})$ signifies the discrete measurement data on $\partial\Omega$, $\|\cdot\|_2$ denotes the $2$-norm in $\mathbb{R}^n$, and $\mu$ is the regularization parameter. Moreover, the measurement data $u^{s,\delta}$ and the exact data $u^{s}$ satisfies $\|u^{s,\delta}-u^{s}\|_2\leq\delta$.

In this paper, we apply the Levenberg-Marquardt method \cite{Doicu2010, Hanke1997} to find the minimizer of $J[q]$, which is essentially a variant of the Gauss-Newton iteration. We suppose that $q^*$ is an approximation of $q$, then the nonlinear mapping $\mathcal{F}$ in (\ref{variational reg}) can be replaced approximatively by its linearization around $q^*$. Thus, minimizing (\ref{variational reg}) can be seen to minimize
\begin{align*}
J[q]:=\frac{1}{2}\parallel \mathcal{F}'(q^*)\delta q-(u^{s,\delta}-\mathcal{F}(q^*))\parallel_2^2+\frac{\mu}{2}\parallel \delta q \parallel_{L^2[0,2\pi]}^2,
\end{align*}
where, $\delta q=q-q^*$.

We know that the regularization method is representative of deterministic inverse solution techniques and can only get point estimates of the solution. To obtain a more statistical description of all possible solutions, we also interpret the inverse problem from a Bayesian perspective. In the classical Bayesian theory, the observation error $\xi$ is assumed to be an independent and identically distributed Gauss random vector with mean zero, the covariance matrix $B=\delta^2 I$ ($I$ is the unit matrix), the minimization functional (\ref{variational reg}) can be rewritten as follows
\begin{align*}
J[q]&\propto \frac{1}{2\delta^2}\parallel \mathcal{F}(q)-u^{s,\delta}\parallel_2^2+\frac{\mu}{2\delta^2}\parallel q\parallel_{L^2[0,2\pi]}^2\\
&=\frac{1}{2}\parallel\mathcal{F}(q)-u^{s,\delta} \parallel_{B}^2+\frac{\eta}{2}\parallel q\parallel_{L^2[0,2\pi]}^2\\
&=:J_B(q),
\end{align*}
where $\eta=\frac{\mu}{\delta^2}$, $\parallel \cdot\parallel_B$ is a covariance weighted norm given by $\parallel\cdot\parallel_{B}=\parallel\delta^{-1}I\cdot\parallel_{2}$. Moreover, the radial function $q(t)$ of its unknown inner boundary curve $\partial D$ is approximated by trigonometric series
\[
q(t)\approx \sum_{k=0}^m a_k\cos kt +\sum_{k=1}^m b_k\sin kt, \ \ \ 0\leq t\leq 2\pi,
\]
where $m\in\mathbb{N}$ and the vector $q=(a_0, ..., a_m, b_1,...,b_m)\in \mathbb{R}^{2m+1}$. The notation $q$ is also used to denote the vector $(a_0, ..., a_m, b_1,...,b_m)^T$, while the minimizer of $J_B$ defines the maximum a posteriori estimator
\begin{align}
\label{MAP}
q_{\mathrm{MAP}}=\arg \min_{q} J_B(q).
\end{align}

In almost all inverse problems, the selection of regularization parameters plays a crucial role. The regularization method heavily depends on the choice of the regularization parameter. We introduce the hierarchical Bayesian model in this paper \cite{WZ2005,KS2005,DZ2021}, which provides an effective way to achieve an automatic and flexible selection of the regularization parameter $\mu$ or $\eta$. The main idea is to treat the regularization parameter as a random variable with its own priors, such as the gamma distribution for $\eta$. Then the posterior density function can be written as
\begin{align}
\label{posterior with hyper}
p(q,\eta|{u^{s,\delta}})\propto \mathrm{exp}\left(-\frac{1}{2}\parallel \mathcal{F}(q)-u^{s,\delta}\parallel_B^2\right)\eta^{(2m+1)/2}\mathrm{exp}(-\frac{\eta}{2}q^{T}q)
\eta^{\alpha_0-1}\mathrm{e}^{-\beta_0\eta},
\end{align}
where $(\alpha_0, \beta_0)$ is the parameter pair of the Gamma distribution.

Then the maximum a posteriori estimator (MAP point) of the posterior density function (\ref{posterior with hyper}) is essentially derived as the minimizer of the functional
\begin{align}
\label{J hyper}
J_B(q,\eta)=\frac{1}{2}\parallel \mathcal{F}(q)-u^{s,\delta}\parallel^2_B+\frac{\eta}{2} q^{T}q +\beta_0\eta-(\frac{2m+1}{2}+\alpha_0-1)\mathrm{ln}\eta.
\end{align}
The functional $J_B(q,\eta)$ in (\ref{J hyper}) is called the augmented Tikhonov regularization. To find the minimizer of augmented Tikhonov regularization and achieve an automatic selection of regularization parameters, we solve the minimizer of (\ref{J hyper}) by alternating iterative, i.e., fixed $\eta_\mathrm{z}$ (respectively, $q_z$), update the $q_{z+1}$ (respectively, $\eta_{z+1}$), and iterate $q_{z}$ using  Levenberg-Marquardt method for each step.

\begin{algorithm}
\small
\caption{alternating iteration method for solving the variational problem (\ref{J hyper}).}
      1:~~Choose $q_0, \eta_0, \alpha_0, \beta_0$, and set $z=0$;\\
      ~2:~~Solve the direct problem $(\ref{sca equ})$ and determine the residual $F_z=u^{s,\delta}-\mathcal{F}(q_z)$;\\
      3:~~Compute the Jacobian matrix $G=\mathcal{F}'(q_z)$;\\
      ~4:~~Calculate $\delta q_z=(G^T G+\eta \delta^2 I)^{-1}(G^T F_z)$;\\
      ~5:~~Update the solution $q_z$ by $q_{z+1}=q_z+\delta q_z$, \\
      ~6:~~Update the parameter $\eta_{z+1}$ by
      \begin{eqnarray*}
      \small
      \label{Qresidual-eq3}
      \begin{split}
      \lambda_{z+1}=\frac{\frac{2m+1}{2}+\alpha_0-1}{q_{z+1}^Tq_{z+1}/2+\beta_0};
      \end{split}
      \end{eqnarray*}
      7:~~Increase $z$ by one and go to step 2, repeat the above procedure until a stopping criterion \\ $~~~~~$is satisfied.\\
   \label{algorithm-pgddf}
\end{algorithm}

The essence of the Laplace approximation is to replace the complicated posterior with the normal distribution located at the maximum posterior value $q_{\mathrm{MAP}}$. In fact, it is a linearization around the $\mathrm{MAP}$ point (cf. \cite{SSP}). It consists of approximating the posterior measure (or distribution) by
$\omega\approx N(q_{\mathrm{MAP}},C_{\mathrm{MAP}})$, where
\begin{align}
\label{C_map}
C_{\mathrm{MAP}}=(J''_B)^{-1}=(\frac{\mu}{\delta^2}I+\frac{1}{\delta^2}G^{T}G)^{-1},
\end{align}
here $G$ is the Jacobian matrix of the forward operator $\mathcal{F}$ at the point $q$. Notice that the covariance formula ($\ref{C_map}$) only uses the first order derivatives of $\mathcal{F}$. The concrete implementation steps of Laplace approximation \cite{Iglesias2013} are given in algorithm 2.

\begin{algorithm}
\small
\caption{Laplace approximation (LA) for sampling.}
      1:~~Compute $q_{\mathrm{MAP}}$ from $(\ref{MAP})$ by using Algorithm 1, and $C_{\mathrm{MAP}}$ from $(\ref{C_map})$, respectively;\\
      ~2:~~Compute the Cholesky factor $L$ of $C_{\mathrm{MAP}}$, i.e., $C_{\mathrm{MAP}}=LL^{T}$;\\
      ~3:~~For $j=\{1,...,N_e\}$, generate $q^{j}=q_{\mathrm{MAP}}+L^{T}\mathcal{B}^{j}$,
      where $\mathcal{B}^{j}\sim N(0,I)$.\\
   \label{algorithm-LA}
\end{algorithm}

\section{ Numerical results and discussions }
This section presents some numerical examples to illustrate the effectiveness and promising features of the proposed reconstruction scheme.

In all of our numerical examples, the measurement boundary curve $\partial\Omega$ is given by the circle of radius 1.5 and centered at the origin, that is $\partial \Omega=\{1.5(\cos t, \sin t ), 0\leq t\leq 2\pi\}$. Set the incident field $u^i=e^{\mathrm{i}k_m x\cdot d}$, where the direction $d=(\cos(\pi/3),\sin(\pi/3))$, and $\mu_m=\varepsilon_m=1, \varepsilon_c=2$. We solve the direct problem with $\mathrm{Nystr\ddot{o}m}$ method \cite{Kress1999}, and there are $2n=50$ grid points.  In addition, to avoid committing an inverse crime, the number of collocation points for obtaining the synthetic data was chosen to be different from the number of collocation points within the inverse solver.

In the iterative process, we choose a circle of radius 1 and centered at the origin as the initial guess. A finite difference method is used to calculate the Jacobian matrix $G$, and the maximum number of iteration steps is set to 100.  The number of samples $N_e$ is 10000 in the algorithm \ref{algorithm-LA}, and the following stopping rule is given
\[
E_z=\parallel q_z-q_{z-1}\parallel_{L^2}\leq 10^{-5}.
\]
The noisy measured data is generated by
\[
u^{s,\delta}=u^s(x)+\delta\frac{\xi}{\mid \xi\mid}, \ \ x\in\partial\Omega,
\]
where $u^s(x)$ is the exact data, $\delta$ indicates the noise level, and $\xi$ is the Gaussian random vector with a zero mean and unit standard deviation.

The approximate solution $\widetilde{q}_{N_e}$ by  LA algorithm compare to the exact solution $q(x)$ by computing the relative error
\[
e_\gamma=\frac{\parallel \widetilde{q}_{N_e}-q(x) \parallel_{L^2(\partial D)}}{\parallel q(x)\parallel_{L^2(\partial D)}}.
\]

\begin{exm}\label{exm-disk}
In this example, we consider the reconstruction of a circle object with
\[
q(t)=0.8, \ \ \ 0\leq t\leq2\pi,
\]
and we take $\alpha_0$=800, $\beta_0=0.01$ and $\eta_0=10$ as the initial values in the iterative algorithm \ref{algorithm-pgddf}.
\end{exm}

First, we investigate the variation of the mode of the scattering field with $\omega$ at a fixed position $x_0=(1.5,0)$ for Example \ref{exm-disk}. It is well known that the eigenvalues of $\mathcal{K}^*_D$ are \{0, 1/2\} for $D$ is a disk. If setting $\lambda_1=0$, from Section 2.2, the quasi-static plasmon resonance is defined by $\omega$ such that $\Re(\lambda(\omega))=\lambda_1$. In Fig.~\ref{examp1-w} (a) is compared the magnitude of the scattering field $|u^s(x_0)|$ for negative materials i.e., $\mu_c=-1+0.004\mathrm{i}$ and $\lambda=0+10^{-3}\mathrm{i}$ with the magnitude of the scattered field $\mu_c=5$ and $\lambda=0.75$ for normal materials. It is clear that as $\omega$ tends to zero, the corresponding scattering field becomes smaller. If $\omega$ is fixed, for negative materials $\mu_c=-1+0.004\mathrm{i}$, that is, the scattering field is always larger than the scattering field for $\mu_c=5$ and $\lambda=0.75$ when resonance occurs. It is concluded that when the resonance occurs, it enhances the intensity of the field, thus increasing the signal-to-noise ratio, which will have a beneficial effect on solving the inverse problem.

In Fig.~\ref{examp1-w}(b), when $\Re(\lambda(\omega))=0$ and the plasmon resonance frequency  $\omega=0.01$, we can see that scattering field gradually increases as the $\Im\lambda(\omega)$ decreases, in particular, the scattering field tends to flatten out after the imaginary part is small with $10^{-4}$. That is, when the loss of magnetic permeability $\mu_c$ is small enough, the scattered field reaches a maximum.

Next, we study the effect of different values $\lambda$ on the reconstruction. In Fig.~\ref{examp1-u} (a) and Fig.~\ref{examp1-u} (b), with  the same error level $\delta=0.01$, the inversion results of $\mu_c=-1+0.004\mathrm{i}$ are better than  $\mu_c=5$. Then, it can be obtained that the reconstruction effect of plasmon resonance is better than that of non-resonance. Even in larger error levels, the reconstructed results of plasmon resonance are still satisfactory, as can be seen in Fig.~\ref{examp1-u} (c) and Fig.~\ref{examp1-u} (d), at error levels $\delta=0.1$ and $\delta=0.3$, respectively.
\begin{figure}[htbp]
\centering
\subfigure[]{
    \label{fig:subfig:a}
  \includegraphics[width=2.3in, height=2.0in]{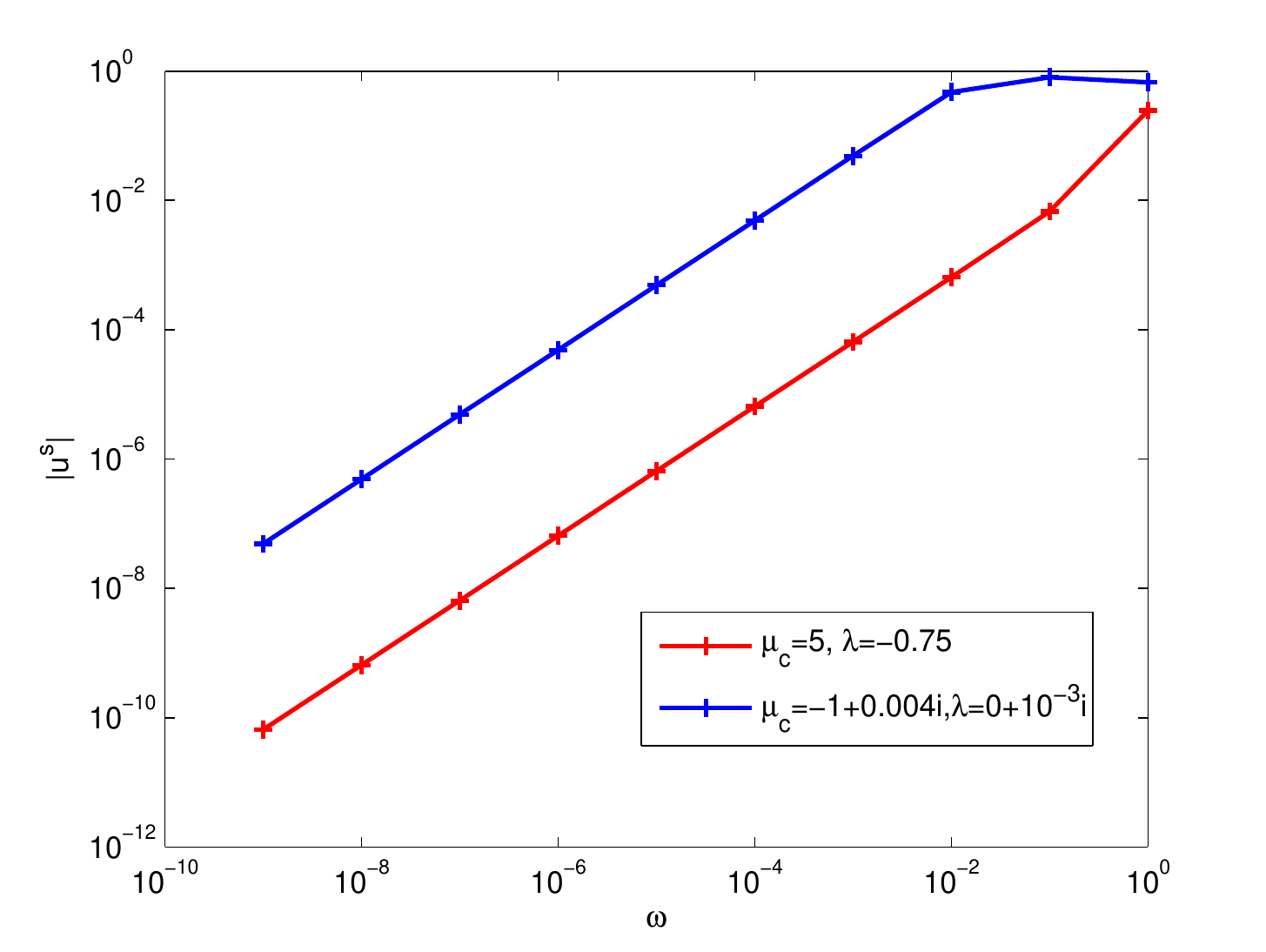}}
  \subfigure[]{
    \label{fig:subfig:b}
  \includegraphics[width=2.3in, height=2.0in]{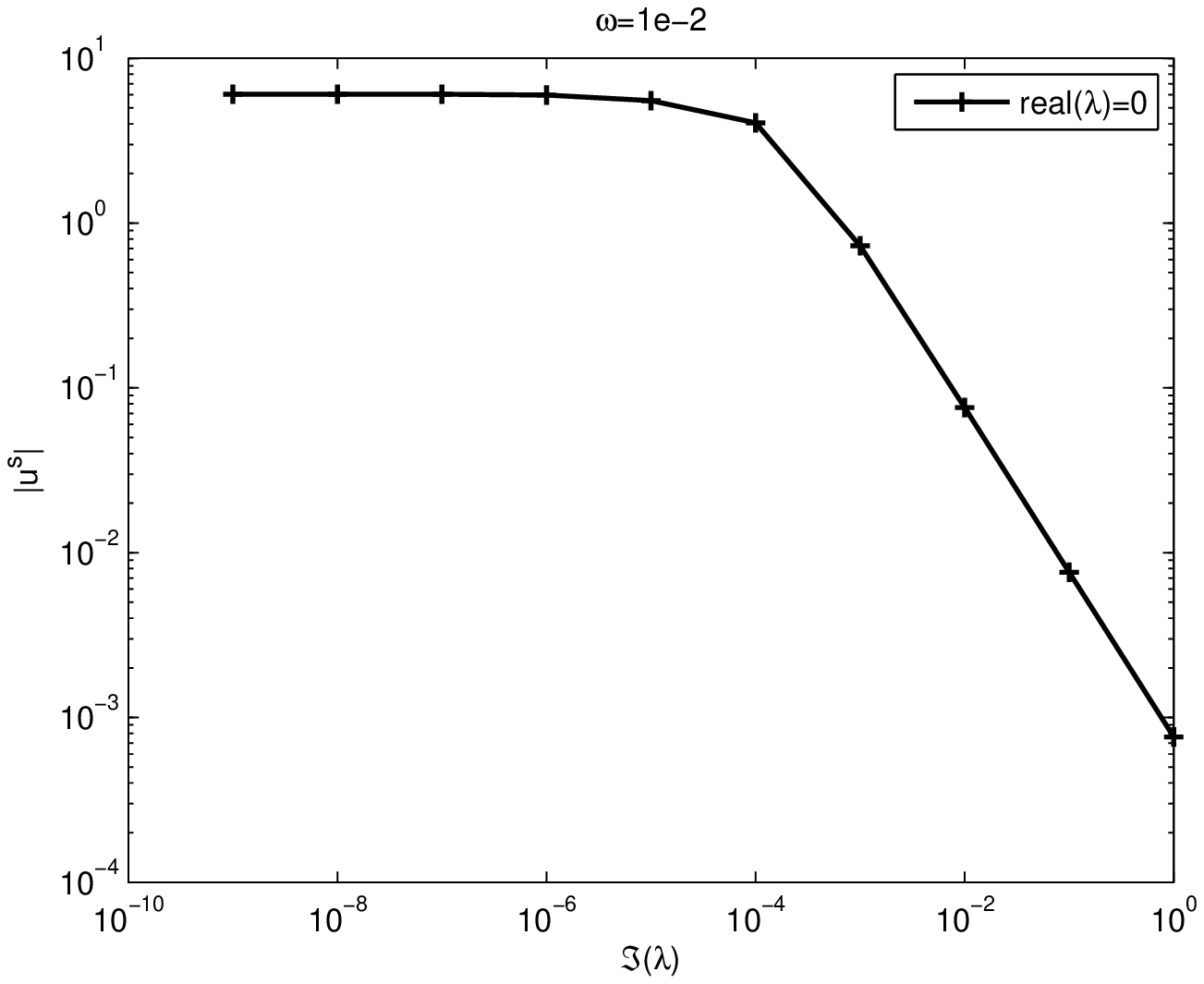}}
  \caption{ (a) The $|u^s|$ in Example \ref{exm-disk} for different $\omega$, (b) when $\Re(\lambda)=0$ and $\omega=0.01$, the result of the variation of $|u^s|$ with $\Im(\lambda)$.}
  \label{examp1-w}
\end{figure}

We explore the variation of $\|SSF(\partial B)\|_{L^2(\partial \Omega)}$  for different sizes of shape scaling factors $\zeta$. We first suppose that $B=\zeta D$ and set $h(x)=x_1+x_2$, then $|\partial B|^{\frac{1}{2}}=\zeta^{\frac{1}{2}}|\partial D|^{\frac{1}{2}}$ and $\|\tau\|_{C(\partial B)}=\frac{1}{\zeta}\|\tau\|_{C(\partial D)}$. In Table \ref{tab1-ssf}, we find that as $\zeta$ increases, the $\|SSF(\partial B)\|_{L^2(\partial \Omega)}$ increases with it. The result is consistent with that of (\ref{diskssf}), and
the upper bound of its $\|SSF(\partial B)\|_{L^2(\partial \Omega)}$ increases with the growth of $\zeta$.
\begin{figure}[htbp]
\centering
\subfigure[$\delta=0.01$]{
    \label{fig:subfig:a}
  \includegraphics[width=2.3in, height=2.0in]{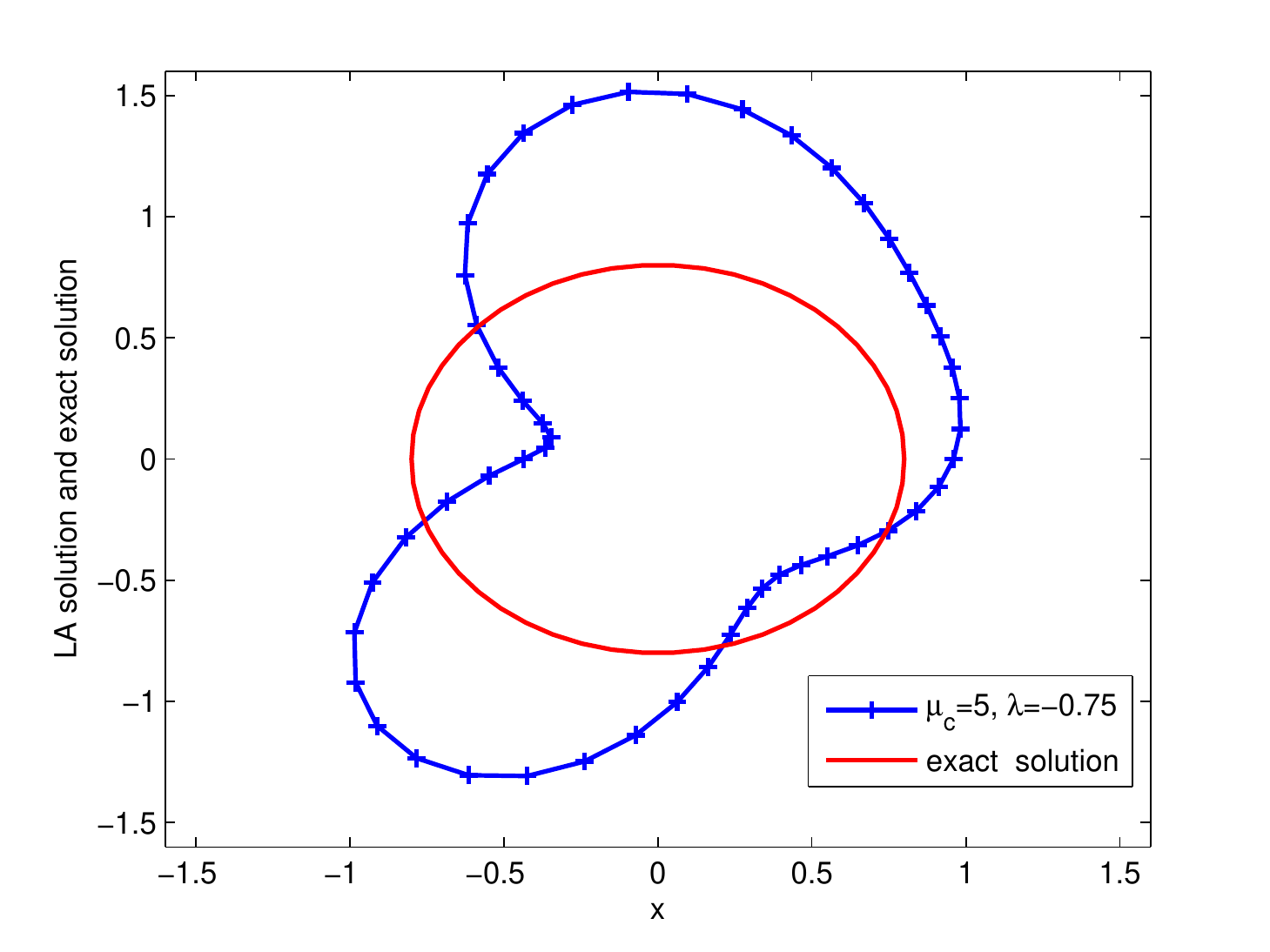}}
  \subfigure[$\delta=0.01$]{
    \label{fig:subfig:b}
  \includegraphics[width=2.3in, height=2.0in]{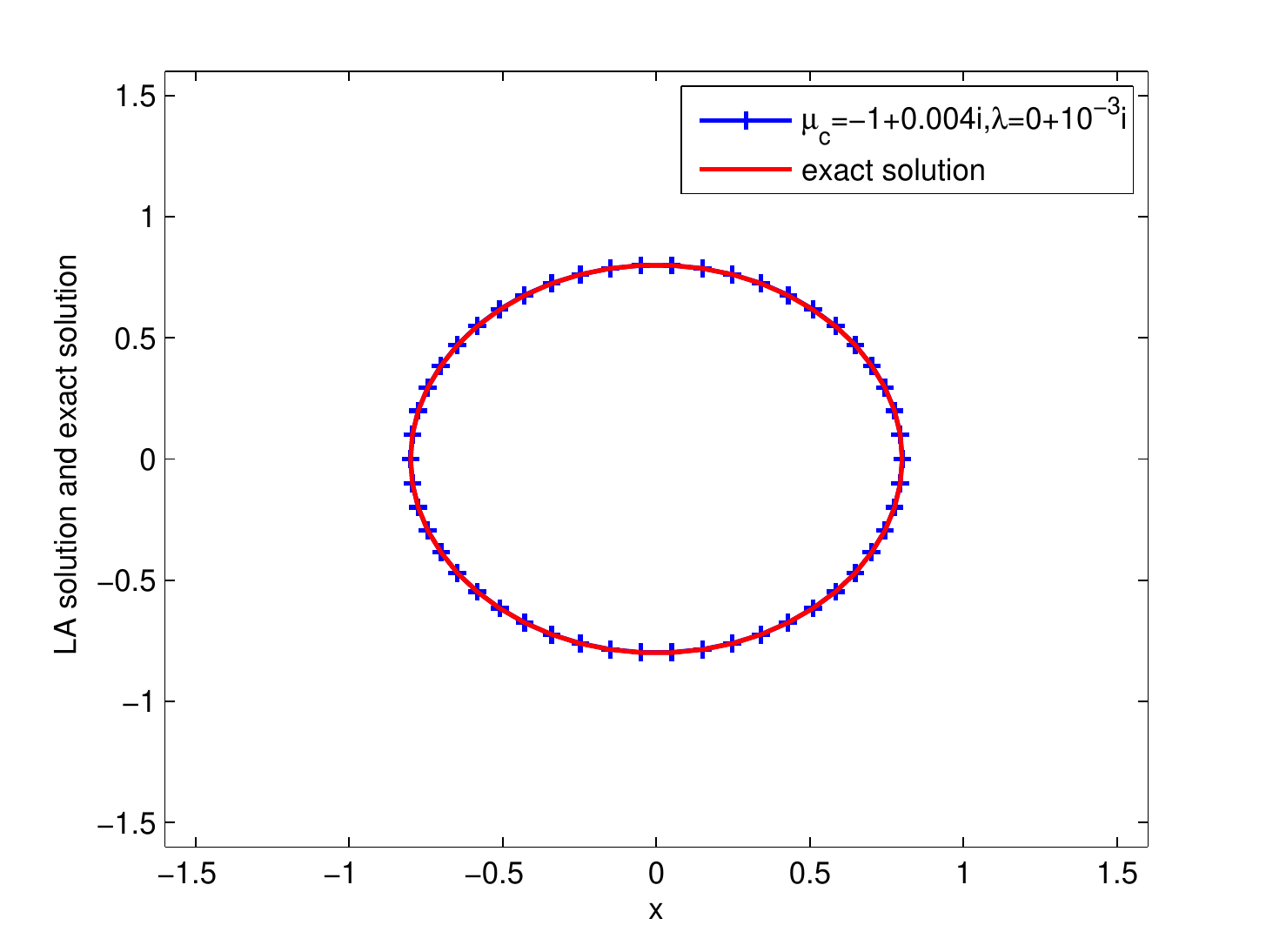}}\\
  \subfigure[$\delta=0.1$]{
    \label{fig:subfig:a}
  \includegraphics[width=2.3in, height=2.0in]{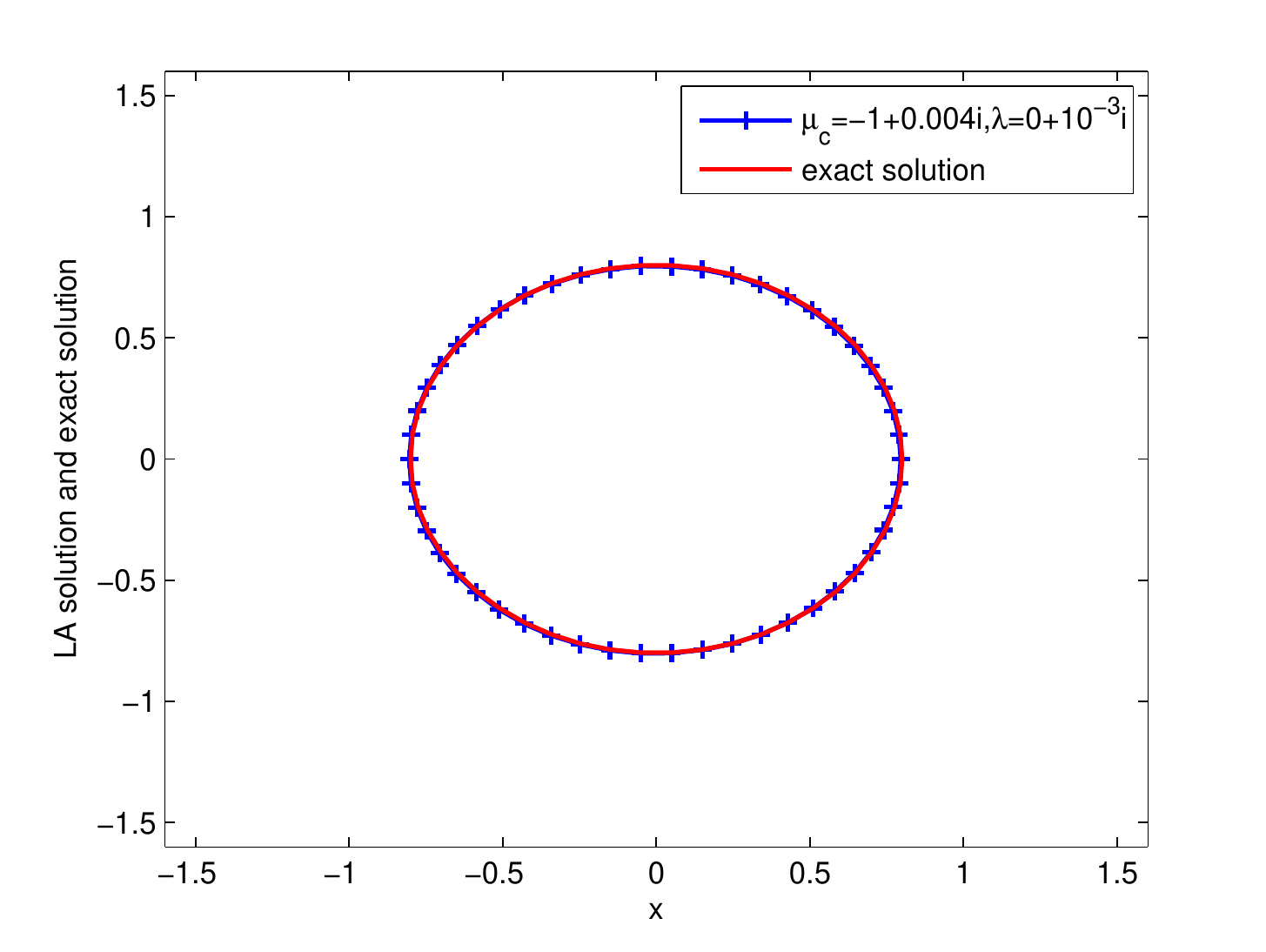}}
  \subfigure[$\delta=0.3$]{
    \label{fig:subfig:b}
  \includegraphics[width=2.3in, height=2.0in]{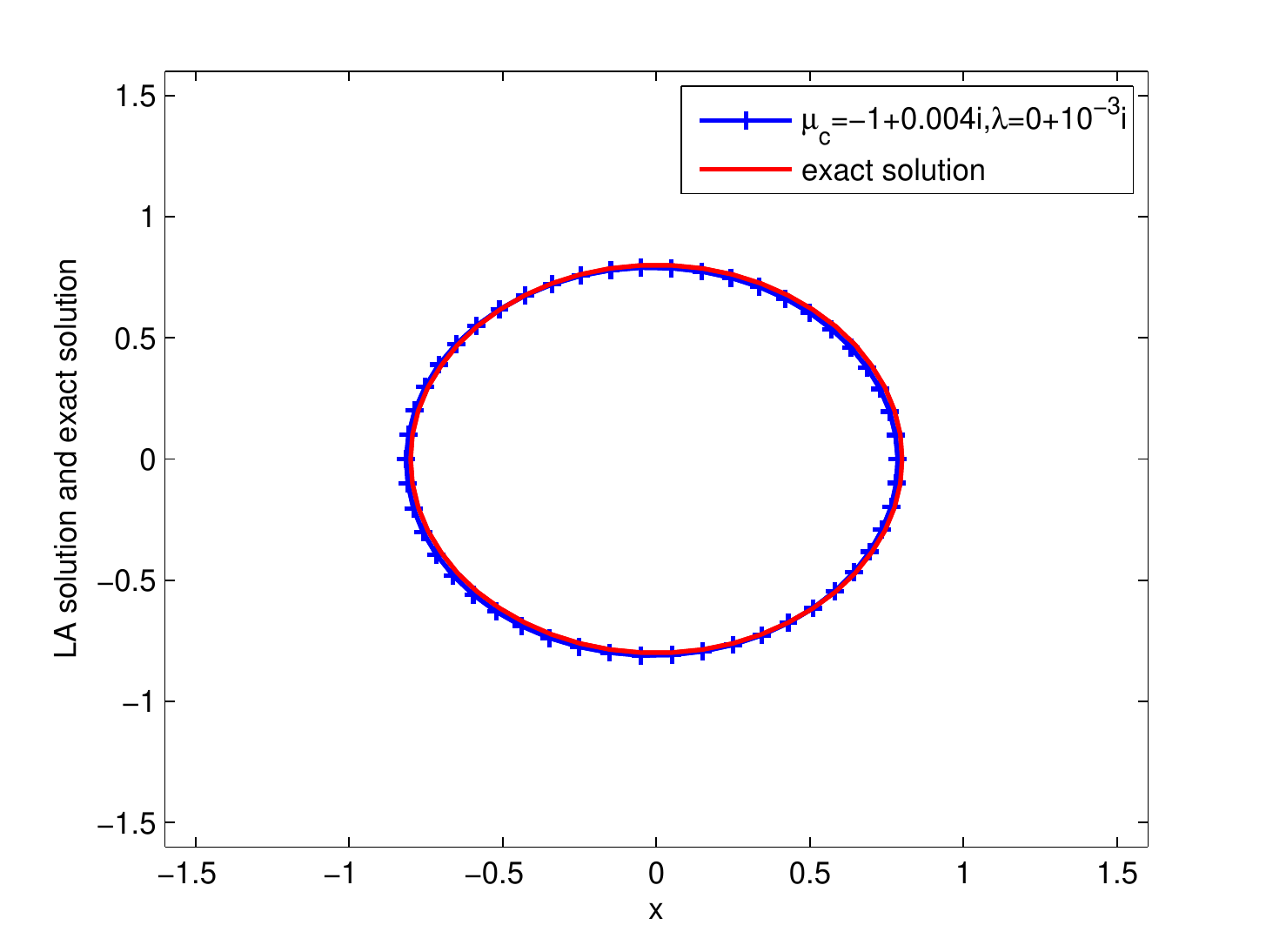}}

  \caption{Reconstruction of the shape in Example \ref{exm-disk} with  different noise level $\delta$. }
  \label{examp1-u}
\end{figure}

\begin{table}[htbp]
\centering
\caption{ The $\|SSF(\partial D)\|_{L^2(\partial \Omega)}$ of Example \ref{exm-disk} for  different $\zeta$. }
\begin{tabular}{c|c|c|c}
\Xhline{1pt}
  $\zeta$ &$|\partial B|^{\frac{1}{2}}$ &$\|\tau\|_{C(\partial B)}$ &$\|SSF(\partial D)\|_{L^2(\partial \Omega)}$\\
  \hline
 $0.50$ &$1.5852$ &$2.500$ &$2.293$\\
  \hline
  $0.67$  &$1.8306$ & $1.875$ &$3.103$  \\
 \hline
 $1$ &$2.2417$&$1.250$& $5.0$ \\
  \hline
  $1.1$ &$2.3514$& $1.136$& $5.698$ \\
  \hline
  $1.2$ &$2.4558$& $1.041$& $6.487$ \\
  \hline
\Xhline{1pt}
\end{tabular}
\label{tab1-ssf}
\end{table}

\begin{exm}\label{exm2}
In the second example, we consider that the inclusion is a peanut, and the polar radius of the peanut is parameterized by
\[
q(t)=\sqrt{\cos^2 t+0.26\sin^2 (t+0.5)},\ \ \ 0\leq t\leq2\pi,
\]
where we take $\alpha_0$=1000, $\beta_0=0.01$ and $\eta_0=1000$ as the initial values in the iterative algorithm \ref{algorithm-pgddf}.
\end{exm}

In Fig.~\ref{examp2-001}, for different $\mu_c=5$ (normal material) and $\mu_c=-0.4508+0.1058\mathrm{i}$ (plasmon resonance materials), it can be seen that the $|u^s|_{\partial \Omega}$ is one orders of magnitude higher than the normal material for Example \ref{exm2} when plasmon resonance occurs. Then we show the reconstruction results for different $\lambda$ and $\mu_c$ in Fig.~\ref{examp2-lamt}.  When $D$ is peanut and the plasma resonance frequency $\omega$ is 0.01, we take three different spectrum $\lambda_1=0.1856, \lambda_2=0.0393 ,\lambda_3=0.003$, respectively, i.e. $\Re(\lambda(\omega))=\lambda_1, \mu_c=-0.4508+0.1058\mathrm{i}$, $\Re(\lambda(\omega))=\lambda_2, \mu_c=-0.8541+0.0034\mathrm{i}$ and
$\Re(\lambda(\omega))=\lambda_3, \mu_c=-0.9392+0.3084\mathrm{i}$. It can be obtained that the inversion effect using plasmon resonance techniques is better than that of non-plasmon resonance of $\mu_c=5, \lambda=-0.75$, and even as the error level increases, the reconstruction remains perfect, see Table \ref{tab1}.

It follows from \cite{Banks2009} that the singular value decomposition of the sensitivity matrix plays a key role in uncertainty quantification. Let the singular value decomposition (SVD) of the sensitivity matrix (Jacobian matrix) $G$ of the forward operator at the true solution $q_{true}$ be denoted as
\begin{eqnarray*}
\label{Qresidual-eq}
\begin{split}
G(q_{true})=U \mathrm{diag}(s_i) V^T
\end{split}
\end{eqnarray*}
with strictly positive decreasing singular values $s_i$, and $U=[U_1\ U_2]$ is an $n\times n$ orthogonal matrix, with $U_1$ containing the first $2m+1$ columns of $U$ and $U_2$ containing the last $n-(2m+1)$ columns. The matrix $V$ is an $(2m+1)\times (2m+1)$ orthogonal matrix, $v_i$ and $u_i$ denote the $i$th columns of $V$ and $U$, respectively. Then the estimator $q$ has the following form:
\begin{align}
\label{SVD}
q=q_{true}+V\Lambda^{-1}U_1^T\widetilde{\xi}=q_{true}+\sum_{i=1}^{2m+1}\frac{1}{s_i}v_iu_i^T\widetilde{\xi}.
\end{align}
From $(\ref{SVD})$, it can be seen that the instability of the inverse problem is caused by the small singular values. In Example \ref{exm2}, the singular values of the sensitivity matrix are calculated at different values of $\lambda$ in Fig.~\ref{examp2-s}. The singular values of the sensitivity matrix where plasmon resonance occurs are all larger than the singular values without resonance. We can see that plasmon resonance can correct the small singular values of the sensitivity matrix $G$, thus effectively reducing the instability of the solution.

Plasmon resonances occurs to enhance the scattering field, improving the signal-to-noise ratio and as $dist(\lambda,\sigma(\mathcal{K}_D^*))\rightarrow 0$, then $\| G\|_{2}=s_{max}$ tends to blow up, see Table \ref{tab2}, which is consistent with the conclusion of Theorem \ref{thm3.4}. However, the condition number of sensitivity matrix $G$  becomes worse with enhanced plasmon resonance (reducing the imaginary part of $\lambda$, making the $dist(\lambda,\sigma(\mathcal{K}_D^*))\rightarrow 0$ ). In the iterative algorithm to solve the matrix of $(G^TG + \eta\delta^2 I)^{-1}$, the condition number of $G$ is too large, then the Fisher information matrix $(G^TG)$ inherits a large condition number, even if the penalty item is increased, it will be invalid, and it cannot alleviate the bad condition number of Fisher information matrix, which leads to the failure of the reconstructed shape.

\begin{figure}[htbp]
\centering
\subfigure[]{
    \label{fig:subfig:a}
  \includegraphics[width=2.3in, height=2.0in]{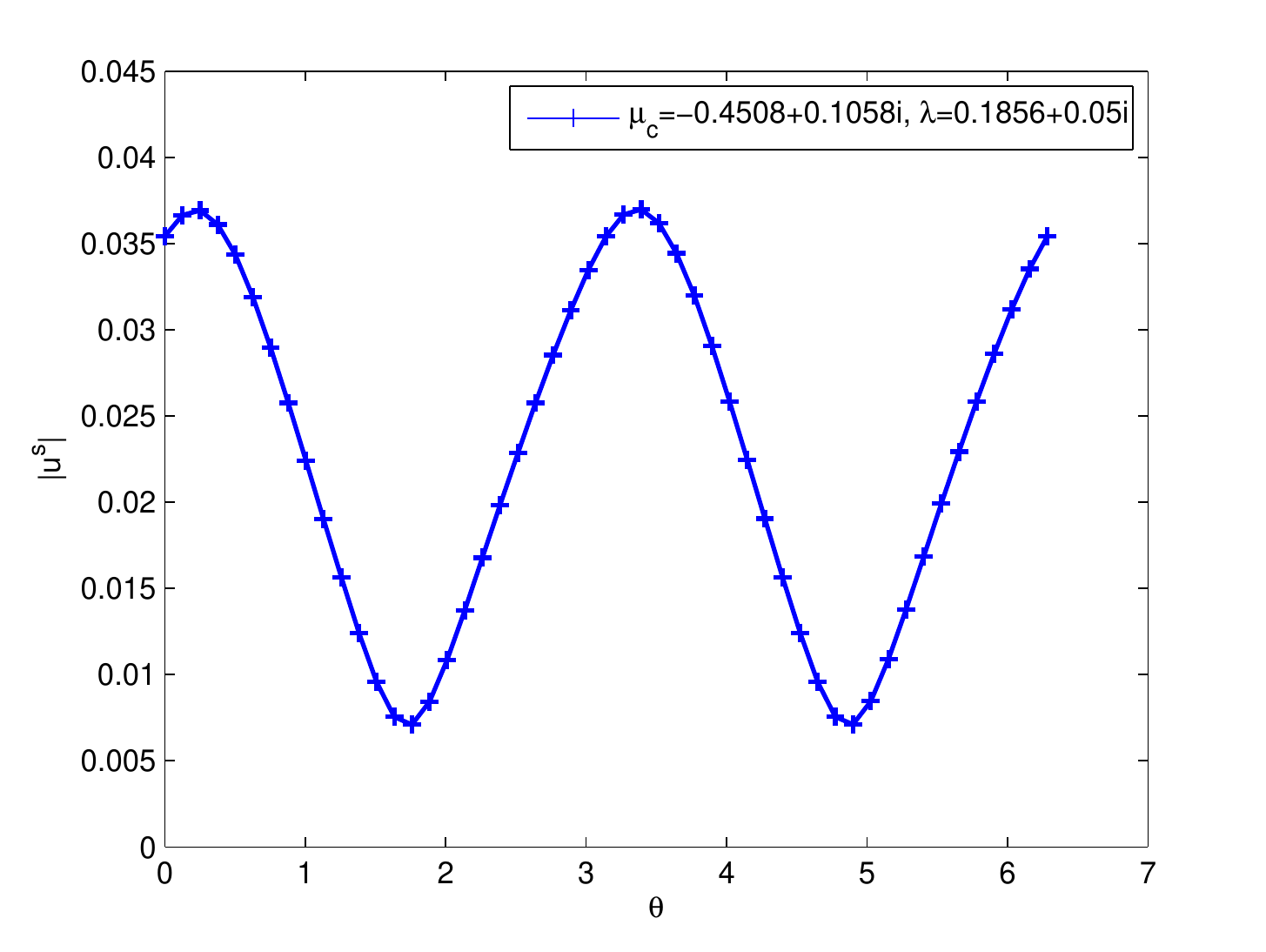}}
  \subfigure[]{
    \label{fig:subfig:b}
  \includegraphics[width=2.3in, height=2.0in]{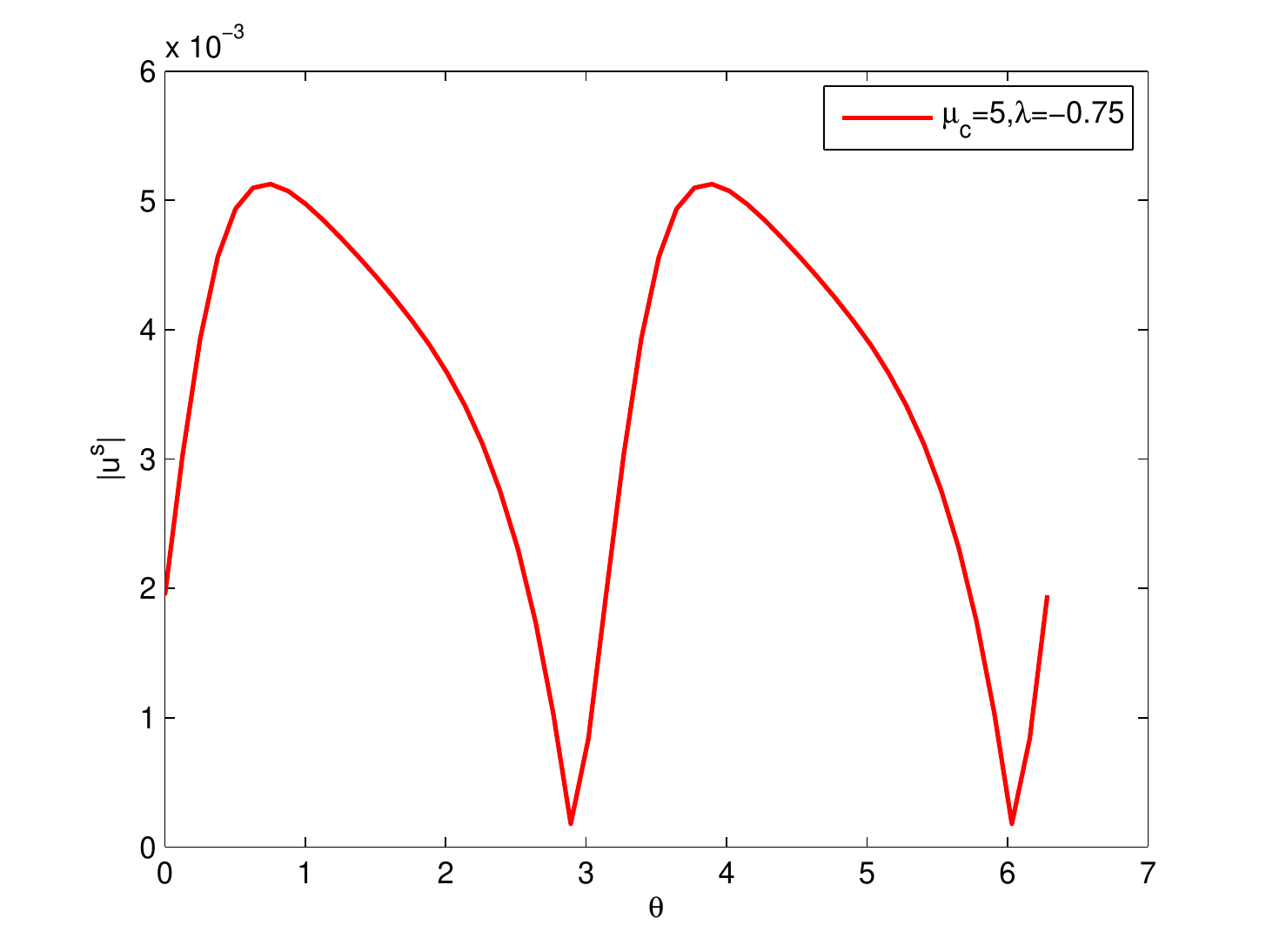}}
  \caption{The $|u^s|$ on a circle of radius 1.5 in Example \ref{exm2} for different $\mu_c$ (a) $\mu_c=-0.4508+0.1058\mathrm{i}$, (b) $\mu_c=5$. }
  \label{examp2-001}
\end{figure}

\begin{figure}[htbp]
\centering
\subfigure[$\delta=0.001$]{
    \label{fig:subfig:a}
  \includegraphics[width=2.3in, height=2.0in]{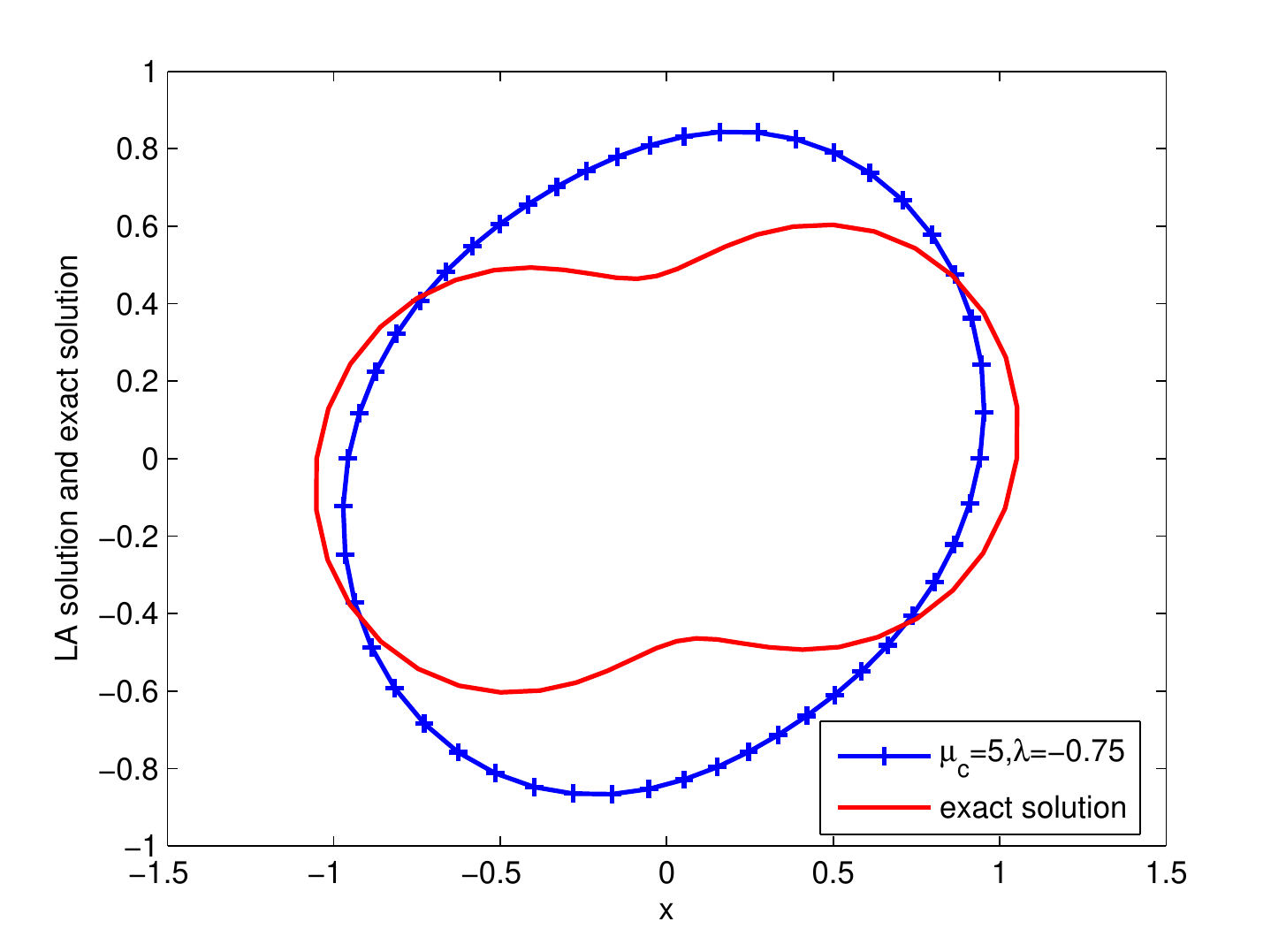}}
  \subfigure[$\delta=0.001$]{
    \label{fig:subfig:b}
  \includegraphics[width=2.3in, height=2.0in]{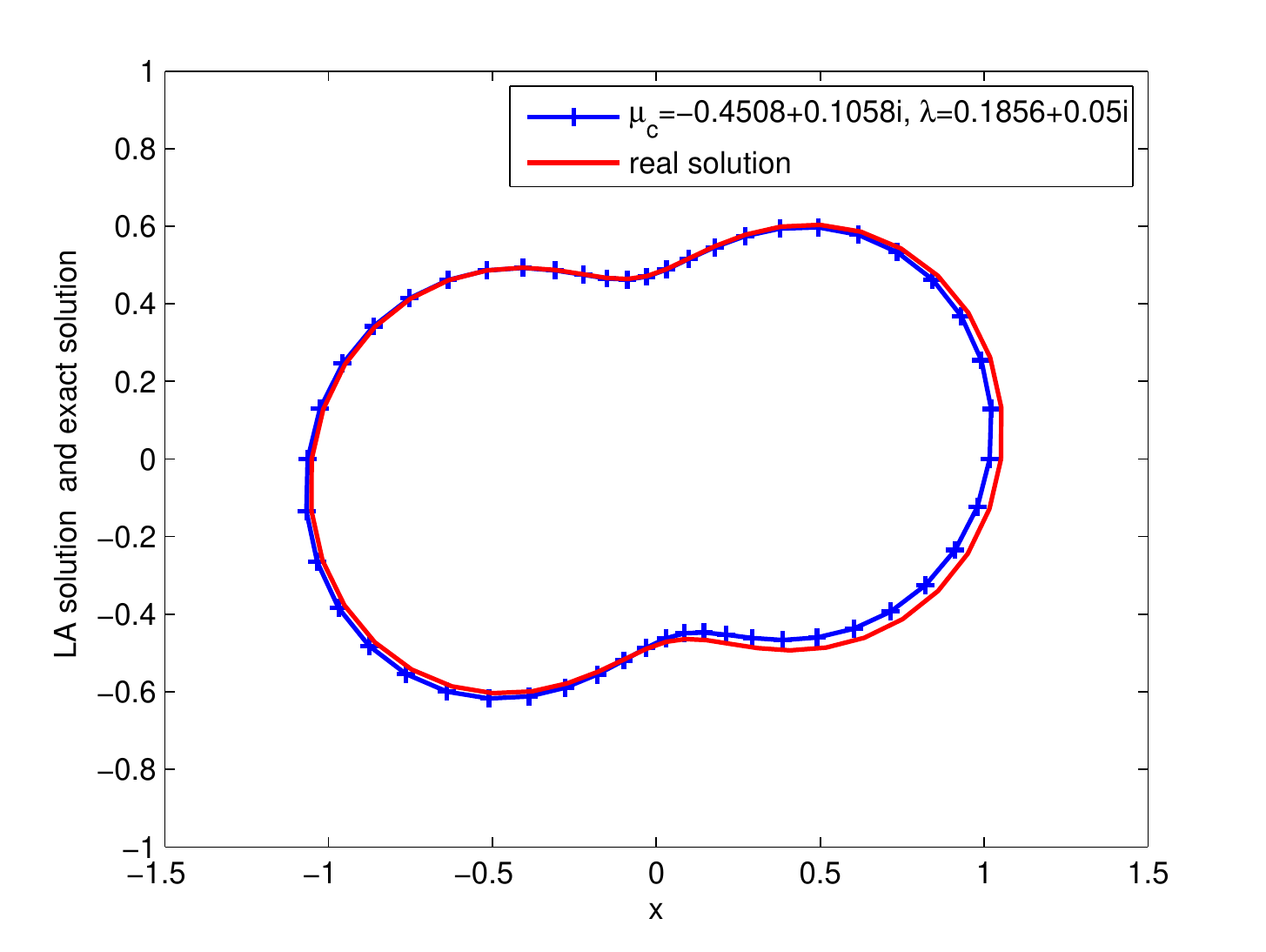}}\\
  \subfigure[$\delta=0.001$]{
    \label{fig:subfig:a}
  \includegraphics[width=2.3in, height=2.0in]{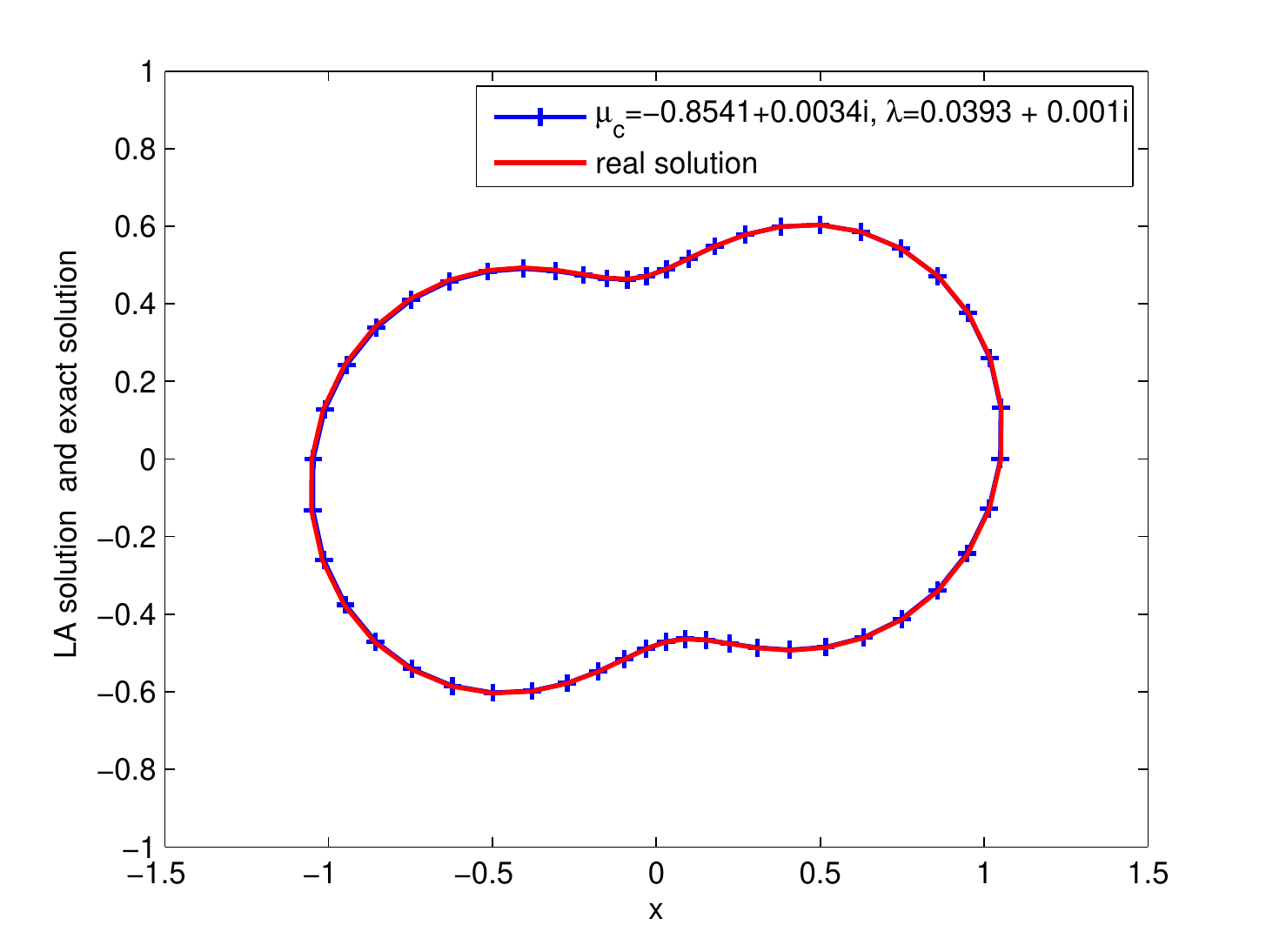}}
  \subfigure[$\delta=0.001$]{
    \label{fig:subfig:b}
  \includegraphics[width=2.3in, height=2.0in]{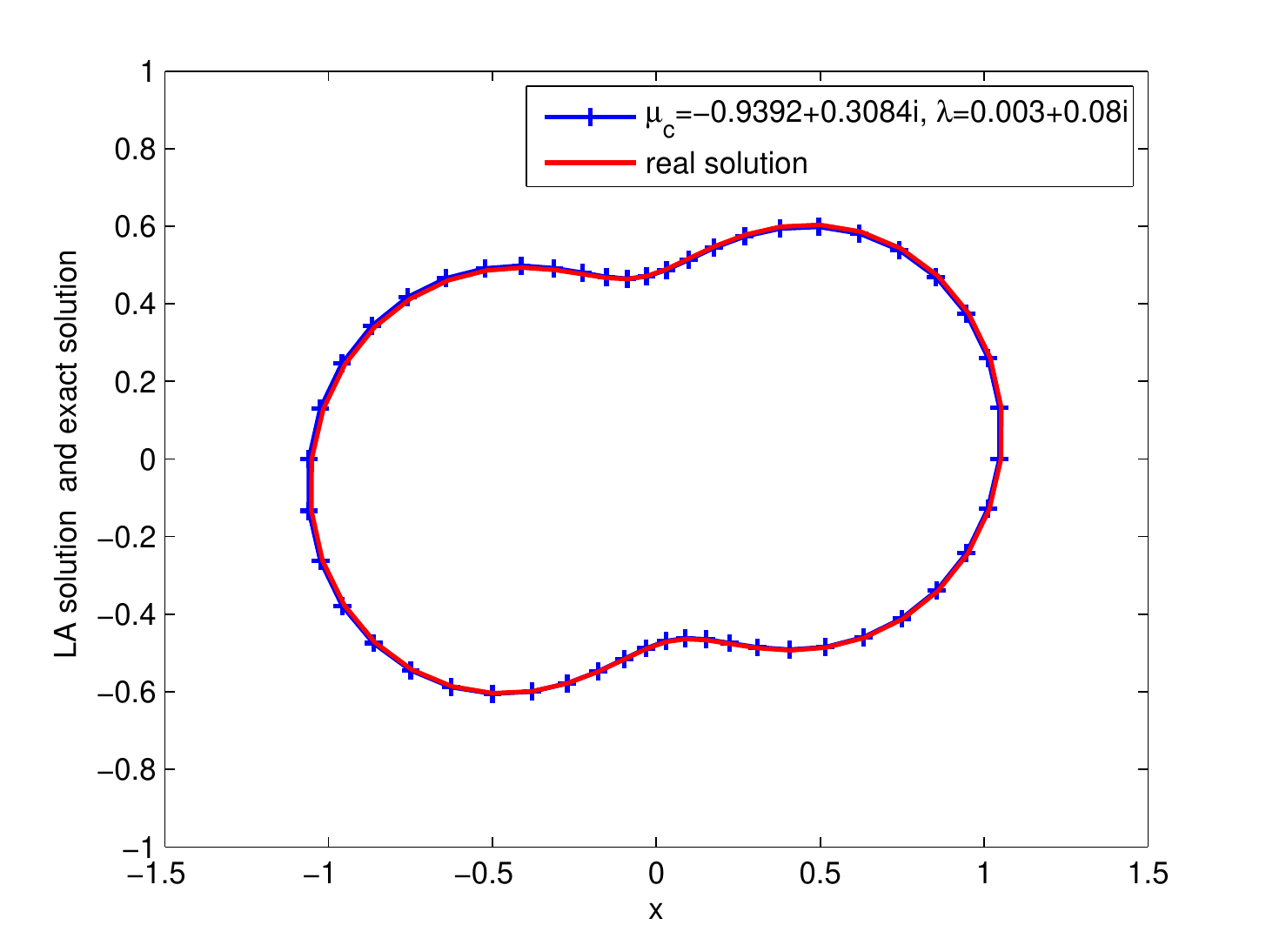}}
  \caption{Reconstruction of the shape in Example \ref{exm2} with 0.001 noise data for different $\mu_c$ or $\lambda$.  }
  \label{examp2-lamt}
\end{figure}
\begin{table}[htbp]
\centering
\caption{ Numerical results of Example \ref{exm2} for $e_\gamma$ associated with different $\lambda$ and noise level $\delta$ }
\begin{tabular}{c|c|c|c}
\Xhline{1pt}
 $\lambda$ &$ \delta=0.001$ &$\delta=0.005$ &$\delta=0.01$\\
  \hline
  $-0.75$  &$0.3179$ &$0.3610$ &$0.4006$  \\
 \hline
 $0.0393+10^{-3}\mathrm{i}$ &$0.0162$&$0.0588$& $0.073$ \\
  \hline
\Xhline{1pt}
\end{tabular}
\label{tab1}
\end{table}
\begin{figure}[htbp]
\centering
\subfigure[]{
    \label{fig:subfig:a}
  \includegraphics[width=2.3in, height=2.0in]{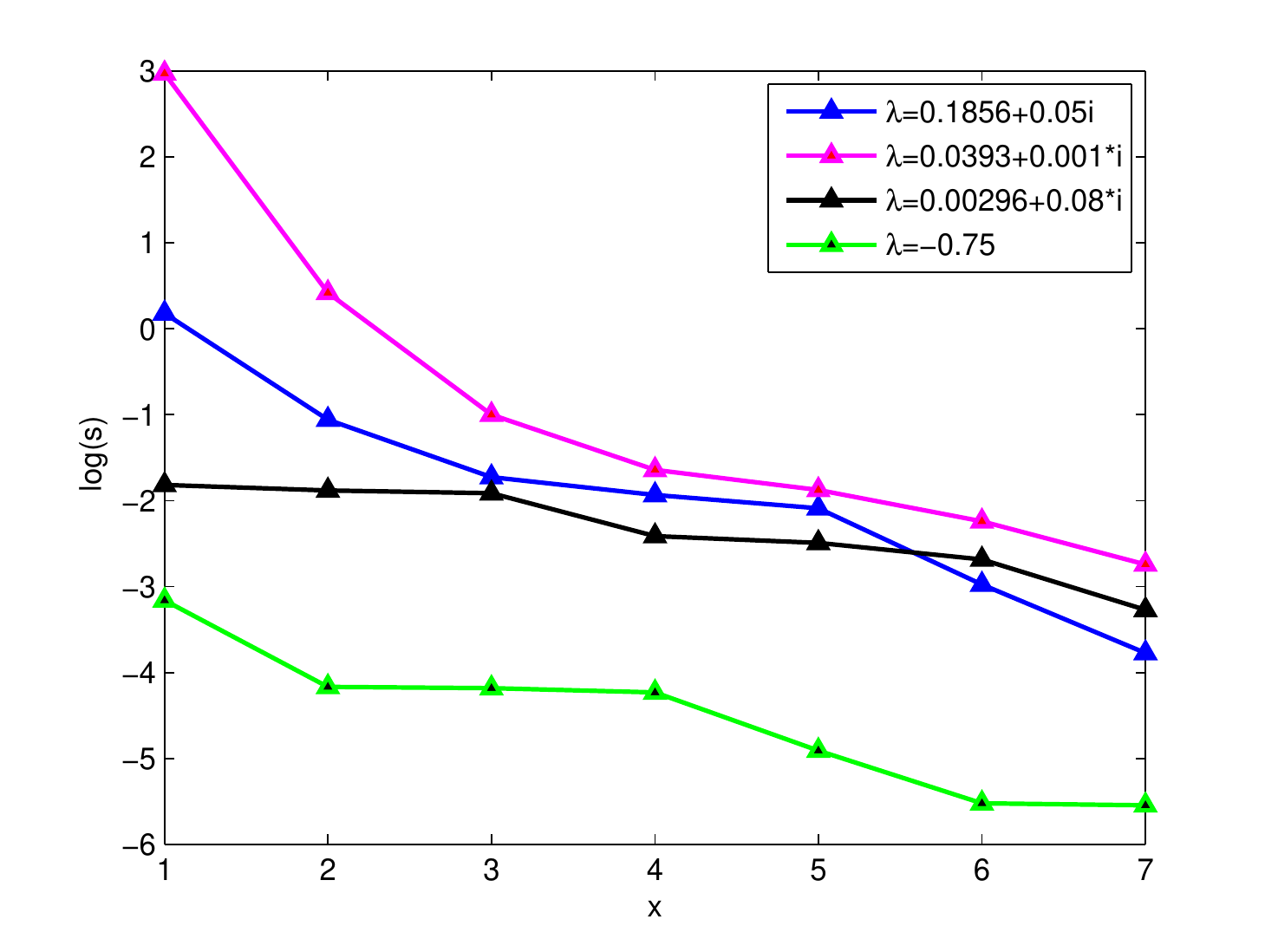}}
  \caption{The distribution of the singular values associated to different values of $\lambda$ for Example \ref{exm2}.}
  \label{examp2-s}
\end{figure}
\begin{table}[htbp]
\centering
\caption{ Numerical results of Example \ref{exm2} for  singular values and condition number of $G$ associated with different $\lambda$. }
\begin{tabular}{c|c|c|c}
\Xhline{1pt}
  $\lambda$ &$s_{max}$ &$s_{min}$ &$cond$\\
  \hline
 $-0.75$ &$0.042$ &$0.0014$ &$28$\\
  \hline
  $0.1856+10^{-1}\mathrm{i}
  $  &$0.289$ & $0.0203$ &$14.25$  \\
 \hline
 $0.1856+10^{-2}\mathrm{i}$ &$30.66$&$0.025$& $1.2\times 10^{3}$ \\
  \hline
  $0.1856+10^{-3}\mathrm{i}$ &$2.9\times 10^3$& $0.070$& $4.2\times 10^{4}$ \\
  \hline
  $0.1856+10^{-4}\mathrm{i}$ &$1.85\times 10^5$& $0.500$& $3.69\times 10^{5}$ \\
  \hline
\Xhline{1pt}
\end{tabular}
\label{tab2}
\end{table}

\begin{exm}\label{exm3}
In this example, we consider that the inclusion is a peach, and the polar radius of the peanut is parameterized by
\[
q(t)=18/25-1/5\sin(t)-3/35\cos(3t),\ \ \ 0\leq t\leq2\pi,
\]
where we take $\alpha_0$=1000, $\beta_0=0.01$ and $\eta_0=1000$ as the initial values in the iterative algorithm \ref{algorithm-pgddf}.
\end{exm}
\begin{figure}[htbp]
\centering
\subfigure[$\mu_c=-0.6618+0.1390\mathrm{i}$]{
    \label{fig:subfig:a}
  \includegraphics[width=2.3in, height=2.0in]{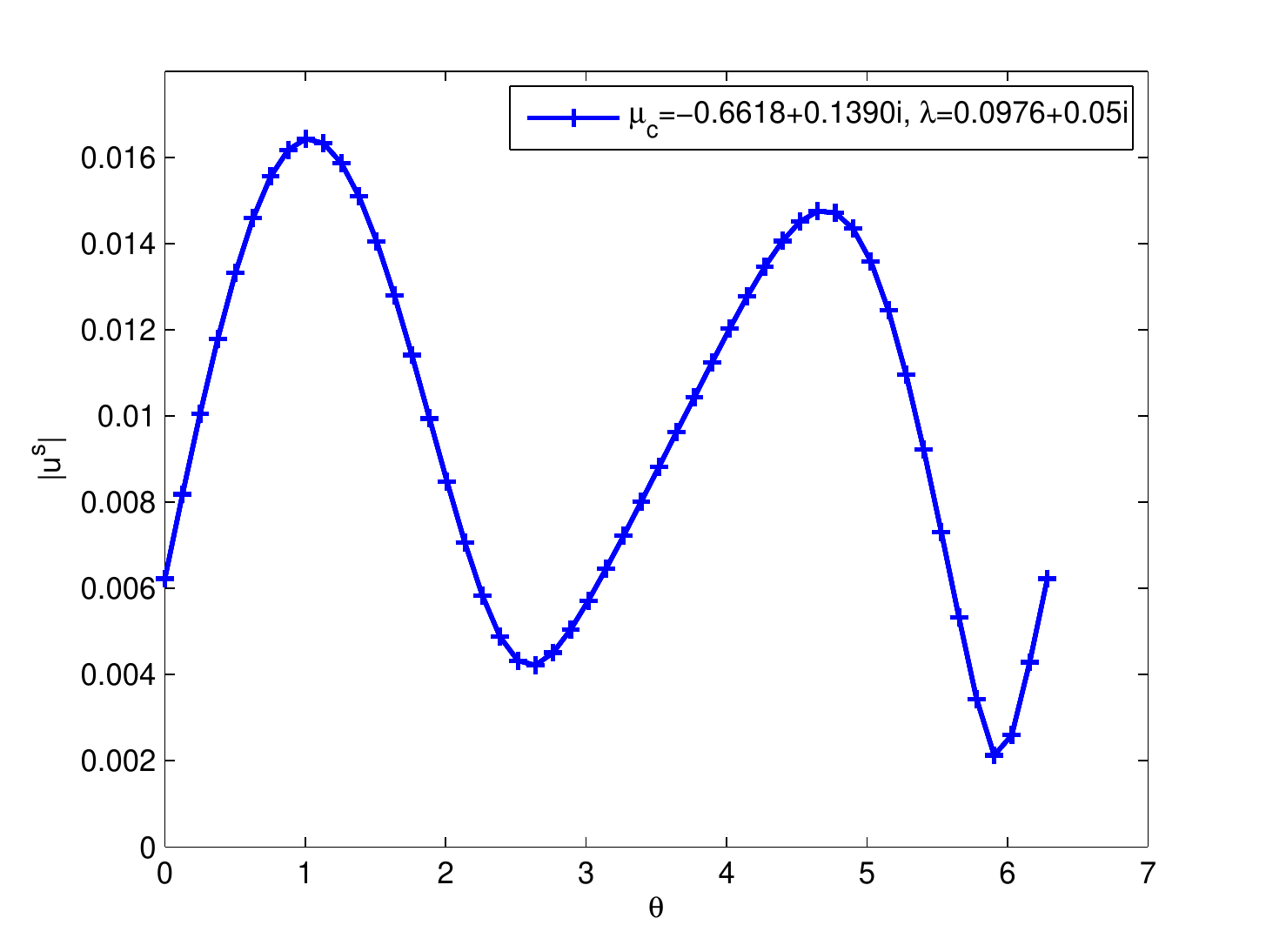}}
  \subfigure[$\mu_c=5$]{
    \label{fig:subfig:b}
  \includegraphics[width=2.3in, height=2.0in]{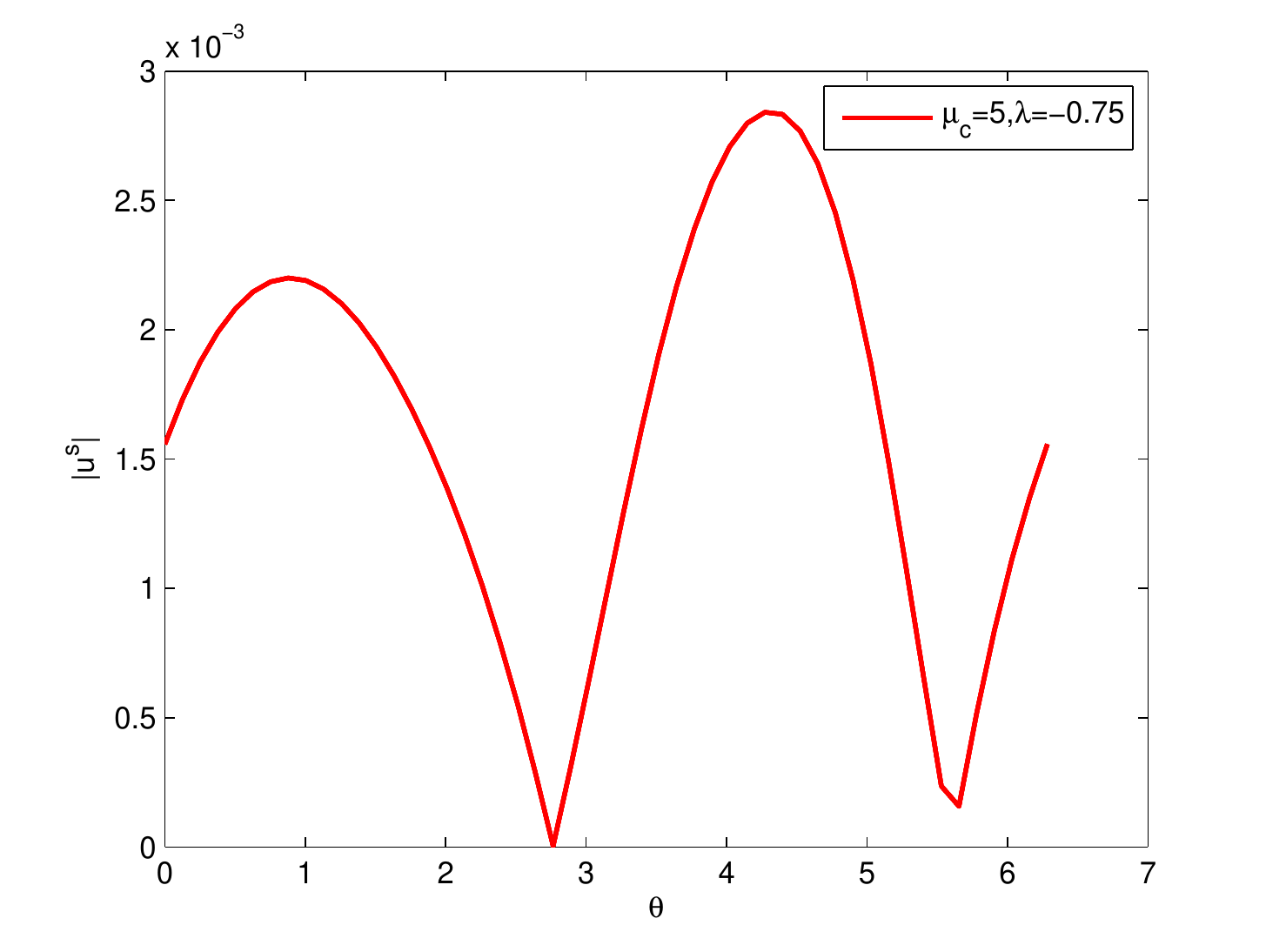}}
  \caption{The $|u^s|$ on a circle of radius 1.5 in Example \ref{exm3} for different $\mu_c$. }
  \label{examp3-u}
\end{figure}

In Example \ref{exm3}, we consider reconstructing a more challenging pear-shaped inclusion and we can reach similar conclusions to the previous two Examples. When plasmon resonance occurs, we get a significant enhancement of the scattering field in Fig.~\ref{examp3-u} and $\widetilde{q}_{N_e}$ coincides well with the exact solution, see Fig.~\ref{examp3-001}. Next, we give the $\mu_c=-0.7372+0.1521\mathrm{i}, \lambda=0.0712+0.05\mathrm{i}$ confidence intervals in Fig.~\ref{confidence}, where the blue region represent the corresponding 95\% confidence regions. In Table \ref{tab4}, the variation of the norm $\|SSF(\partial D)\|_{L^2(\partial \Omega)}$ with different $\zeta$ is presented, and it is obtained that as the scale factor $\zeta$ grows, $|\partial D|^{\frac{1}{2}}$ increases, the curvature of $\|\tau\|_{C(\partial D)}$ becomes smaller, and the norm $\|SSF(\partial D)\|_{L^2(\partial \Omega)}$ grows larger.

\begin{figure}[htbp]
\centering
\subfigure[$\delta=0.001$]{
    \label{fig:subfig:a}
  \includegraphics[width=2.3in, height=2.0in]{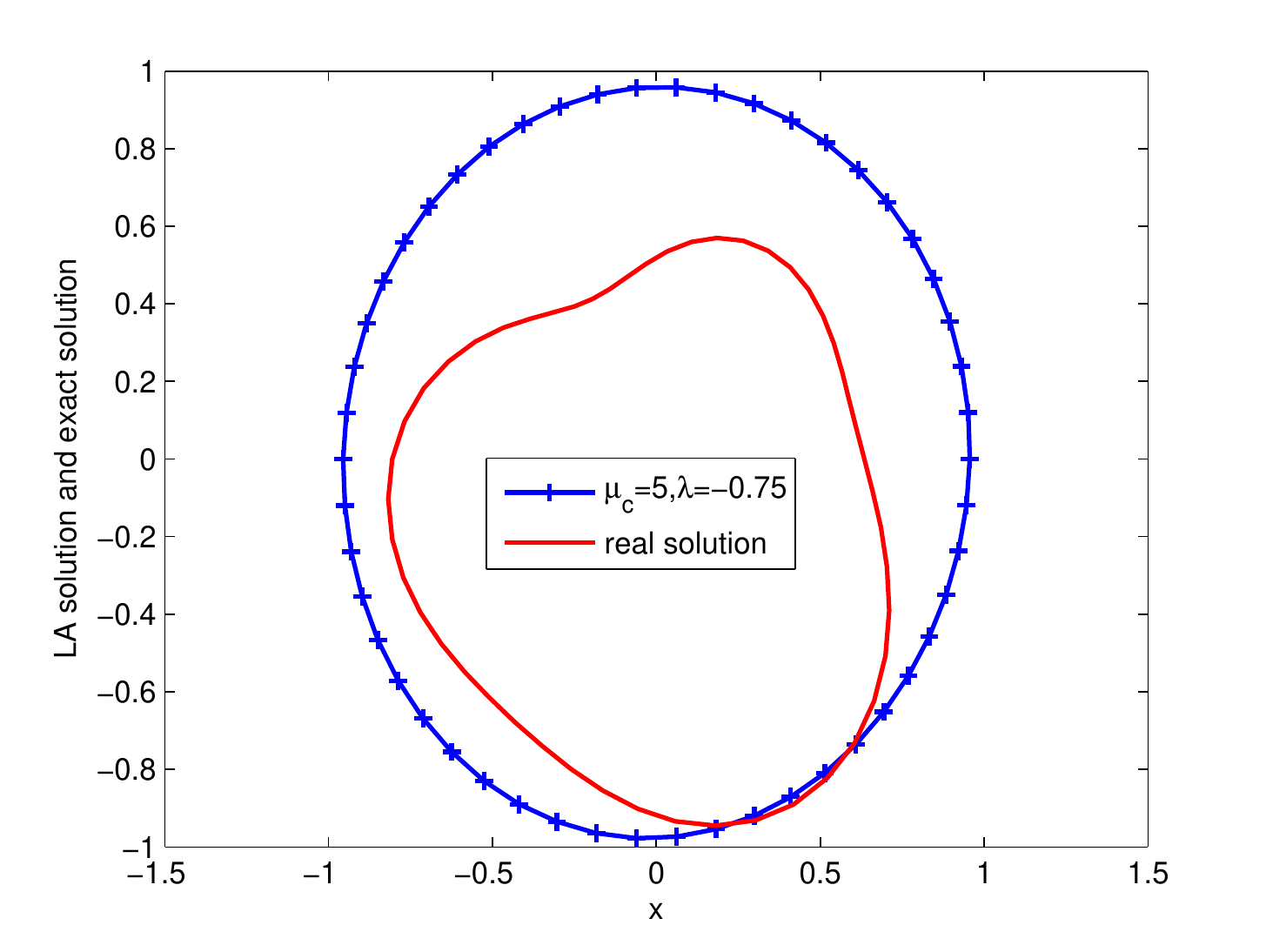}}
  \subfigure[$\delta=0.001$]{
    \label{fig:subfig:b}
  \includegraphics[width=2.3in, height=2.0in]{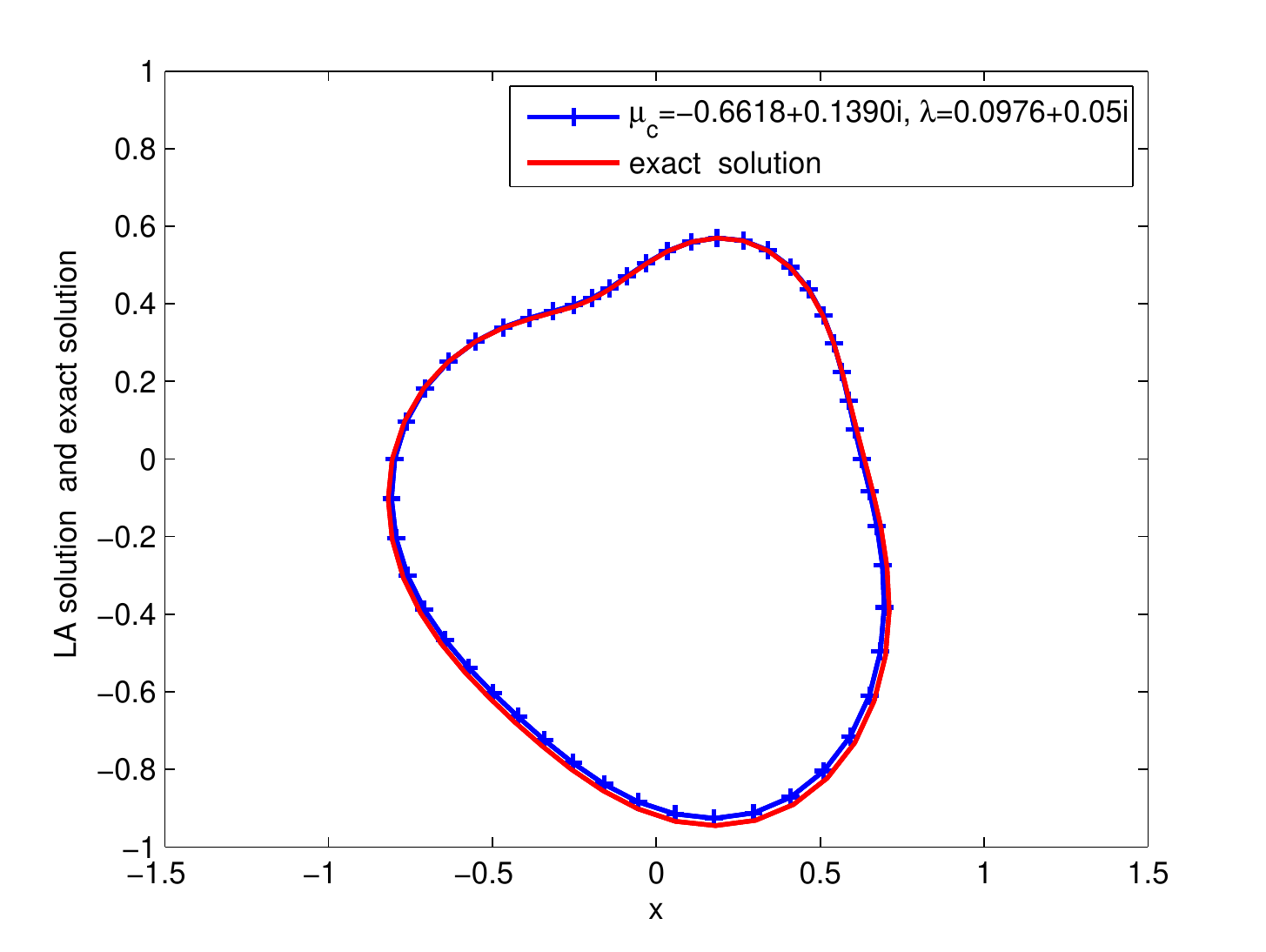}}\\
  \caption{ Shape reconstruction in Example \ref{exm3} with 0.001 noise data for different $\mu_c$ or $\lambda$.}
  \label{examp3-001}
\end{figure}

\begin{figure}[htbp]
  \centering
  \includegraphics[width=2.3in, height=2.0in]{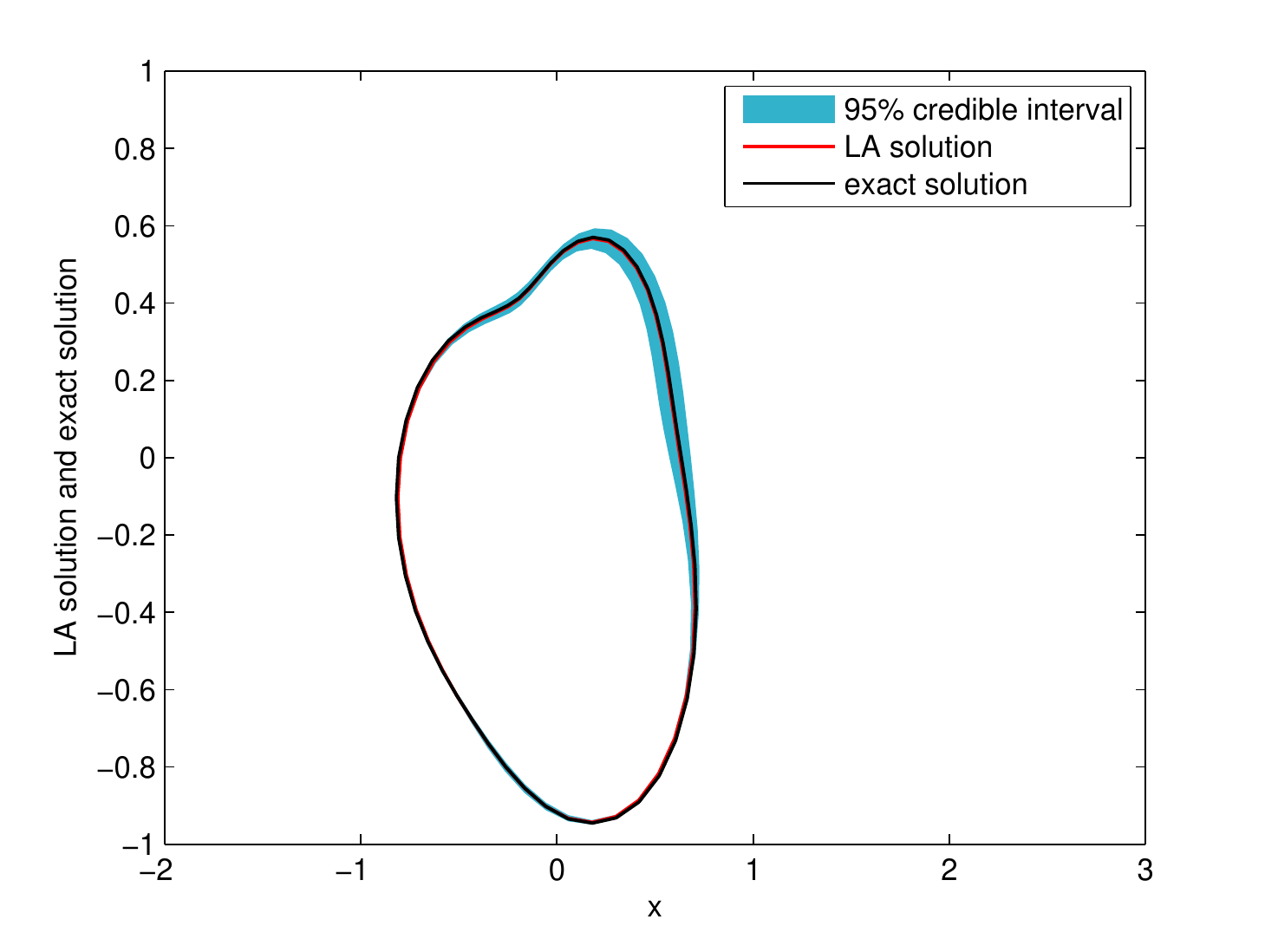}\\
  \caption{The numerical results for Example \ref{exm3} with 0.001 noise data and 95\% confidence interval.}
  \label{confidence}
\end{figure}

\begin{table}[htbp]
\centering
\caption{ The $\|SSF(\partial D)\|_{L^2(\partial \Omega)}$ of Example \ref{exm3} for  different $\zeta$. }
\begin{tabular}{c|c|c|c}
\Xhline{1pt}
  $\zeta$ &$|\partial D|^{\frac{1}{2}}$ &$\|\tau\|_{C(\partial D)}$ &$\|SSF(\partial D)\|_{L^2(\partial \Omega)}$\\
  \hline
 $0.50$ &$1.5628$ &$6.981$ &$0.147$\\
  \hline
  $0.67$  &$1.8044$ & $5.236$ &$0.263$  \\
 \hline
 $1$ &$2.1909$&$3.491$& $0.608$ \\
  \hline
  $1.1$ &$2.3180$& $3.173$& $0.743$ \\
  \hline
  $1.2$ &$2.4210$& $2.909$& $0.895$ \\
  \hline
\Xhline{1pt}
\end{tabular}
\label{tab4}
\end{table}

\section{Conclusions}

In this paper, we study the inverse problem of reconstructing the shape of an anomalous nano-sized inclusion by using the scattering field measurement data in the quasi-static regime. We propose a method utilizing plasmon resonance techniques, to improve the sensitivity of the reconstruction and reduce the ill-posedness of the inverse problem.
First, based on the asymptotic expansion of the layer potential operator, we investigate the asymptotic expansion of the perturbation domain and derive the spectral expansion of the shape sensitivity functional. The results show that the shape sensitivity functional increases rapidly or even blows up with the occurrence of plasmon resonance. Then, to overcome the ill-posedness of the inverse problem, we combine the Tikhonov regularization method with the Laplace approximation to solve the inverse problem based on the hierarchical Bayesian model. This method enables the flexible selection of regularization parameters and obtains statistical information about the solution. Finally, the effectiveness and feasibility of the proposed method are verified by three different examples.

\section*{Acknowledgments}
The work of M Ding and G Zheng were supported by the NSF of China (12271151) and NSF of Hunan (2020JJ4166). The work of H. Liu is supported by the Hong Kong RGC General Research Funds (projects 11311122, 11300821 and 12301420),  the NSFC/RGC Joint Research Fund (project N\_CityU101/21), and the ANR/RGC Joint Research Grant, A\_CityU203/19.

\end{document}